\documentclass{layout}
\usepackage{amssymb}
\usepackage{graphicx}
\usepackage[colorlinks=true, linkcolor=blue, citecolor=webgreen]{hyperref}
\usepackage{color}
\usepackage{bbm}
\usepackage[round]{natbib}
\usepackage{IEEEtrantools} 
\usepackage{tikz}
\usetikzlibrary{matrix}
\newcommand{\R}{\mathbb{R}}
\newcommand{\E}{\mathbb{E}}
\newcommand{\G}{\mathbb{G}}
\newcommand{\N}{\mathbb{N}}
\newcommand{\Z}{\mathbb{Z}}

\newcommand{\C}{\mathcal{C}}

\newcommand{\F}{\mathcal{F}}
\newcommand{\cR}{\mathcal{R}}

\newcommand{\cL}{\mathcal{L}}

\newcommand{\I}{\mathcal{I}}

\renewcommand{\L}{\mathrm{L}}

\renewcommand{\P}{\mathbb{P}}

\newcommand{\X}{\mathcal{X}}

\renewcommand{\d}{\;dt}

\newcommand{\xqed}[1]{%
  \leavevmode\unskip\penalty9999 \hbox{}\nobreak\hfill
  \quad\hbox{\ensuremath{#1}}}
  \renewcommand\qedsymbol{$\blacksquare$}
  \def\ignore#1{}
\def\b1{\mathbbm 1}

\DeclareMathOperator{\sgn}{sgn}
\newtheorem{Theo}{Theorem}[section]
\newtheorem{Lem}[Theo]{Lemma}
\newtheorem{Prop}[Theo]{Proposition}

\theoremstyle{definition}
\newtheorem{Def}[Theo]{Definition}
\newtheorem{Ass}[Theo]{Assumption}

\definecolor{webgreen}{rgb}{0,.5,0}
\theoremstyle{rem}
\newtheorem{rem}[Theo]{Remark}

\numberwithin{equation}{section}

\allowdisplaybreaks
\begin{document}

\title{Numerical solution of  a semilinear parabolic degenerate Hamilton-Jacobi-Bellman  equation with singularity 
} 

\author{Mourad Lazgham}

\thanks{The author acknowledges support by Deutsche Forschungsgemeinschaft through Grant SCHI 500/3-1.}


\dedicatory{\emph{Department of Mathematics, University of Mannheim,  Germany}}


\begin{abstract} 
We consider a semilinear parabolic degenerated \emph{Hamilton-Jacobi-Bellman (HJB) equation with singularity} which is related to a stochastic control problem with fuel constraint. The fuel constraint translates into a singular initial condition for the HJB equation. We first propose a transformation based on a  change of variables that gives rise to an equivalent  HJB equation with nonsingular initial condition but irregular coefficients.
We then construct explicit and implicit numerical schemes for solving the transformed HJB equation and prove their convergences by establishing an extension  to the result of  Barles and Souganidis (1991). 
\\\end{abstract}

\maketitle

\section{Introduction}
In this paper, we aim at  constructing a numerical scheme in order to approximate the solution of a semilinear parabolic degenerate Hamilton-Jacobi-Bellman  equation with singularity, which originates  from an expected utility maximization problem with finite fuel constraint, i.e., where initial and terminal conditions are imposed on the  processes  considered; see, e.g.,  \citet{SST10nn}. This appears to be a very difficult task, since we have to face some issues. Let us in the first place enumerate these ones, theoretically.  First, we cannot directly apply  well-known  convergence results such as in \citet{BS91}, since in their work, they consider only bounded functions with no singularity. Indeed, in most of the literature, when dealing with monotone numerical schemes to approximate Hamilton-Jacobi-Bellman equations, like in \citet{BJ02n} (where they discuss the rate of convergence of approximation schemes), or  more recently, in \citet{ZCB12n} (which is a generalization of the framework of Barles and Souganidis), only\emph{ bounded} viscosity solutions are considered. However,  slight modifications in the Barles and Souganidis framework permit us to adapt their model to viscosity solutions with linear asymptotic growth. Moreover,  a classical change of variables formula will allow us to relax the exponential growth requirement, by introducing an auxiliary HJB equation. Nevertheless, we will still face a polynomial growth and, above all,  a singularity at time $0$, so that to the best of our knowledge no well-known convergence results for monotone schemes can be directly applied in our case. Fortunately, to deal with the singularity property, we will be able to prove that our auxiliary value function behaves like a predetermined function at time $0$, i.e.,  the quotient of the auxiliary value function and this predetermined function will be close to one, near the initial condition. In this manner, we will be able to transform again our auxiliary HJB equation, by considering a translated version of the latter one, which will permit us to set a zero function as initial condition. However, even with our relaxed conditions, classical results for monotone numerical schemes cannot be directly applied here, since there remains a term which behaves like $Tf(X_0/T)$, where $f$ is a strictly convex and positive function with at most  polynomial growth.

 Note that there are other ways to approximate nonlinear parabolic equations. For instance, in \citet{ZB04n}, analyzing  generalized finite difference methods, non-monotone converging schemes are established. In \citet{W13},  the convergence is established for some general approximations of the viscosity solutions, provided that a certain optimization problem can be solved in each time step. Unfortunately, here again only bounded viscosity solutions are considered. An alternative approach to approximate nonlinear parabolic PDEs  would be to use Monte Carlo methods, combined with the finite difference method, as  suggested in \citet{FT11}. In their work, they introduce a backward probabilistic scheme that permits to approximate the solution of a nonlinear PDE in two steps. In the first step, the linear part of the PDE is dealt with   by using Monte Carlo simulation applied to a conditional expectation operator. The second step applies a finite difference method to the remaining nonlinear part. Moreover, they  consider viscosity solutions having polynomial or exponential growth. Nevertheless, the second-order parabolic partial differential equation has to fulfill a Lipschitz condition, uniformly in $t$, which cannot be the case in our framework, due to the Fenchel-Legendre term of the auxiliary HJB equation. In addition, as argued in their paper, their results do not apply to general degenerate nonlinear parabolic PDEs, and we therefore cannot use directly their method. 
 
  In order to remedy to those listed issues, we will have to localize the requirements of building  converging monotone schemes; the fact that our second-order term is one-dimensional will be very helpful to us. However, this will lead to some severe Courant-Friedrichs-Lewy (CFL) conditions in the time parameter and, as a consequence, numerical schemes will converge slowly, since the number of time iterations will have to be chosen sufficiently large.   \\
  \section{Modeling framework}
  Let $(\Omega, \F, \P)$ be a probability space with a filtration $(\F_t)_{0\leq t\leq T}$ satisfying the usual conditions. Taking $X_0\in\R^d$, we consider the following expected utility maximization problem
\begin{equation}
V(T,X_0,R_0)=\sup_{\xi\in\dot{\X}^1(T, X_0)}\E\left[u\left(\cR_{T}^{\xi}\right)\right],\label{omp}
\end{equation}
where
\begin{eqnarray*}
\lefteqn{\dot{\X}^1(T,X_0)}\\
&:=&\Big\{\xi\,\big|\,\Big(X^\xi_t:=X_0-\int_0^t \xi_s\;ds\Big)_{t\in[0,T]}\text{ adapted, } t\rightarrow X^\xi_t(\omega) \in\X_{det}(T, X_0)\, \P\text{-a.s.}\Big\}\\
&\bigcap& \Big\{\xi\,\big|\,\E\bigg[\int_{0}^{T}\big(X^\xi_t\big)^\top\sigma X^\xi_t + |b\cdot X^\xi_t-f(\xi_t)|+|\xi_t|\,dt\bigg]<\infty\Big\},
\end{eqnarray*}
\begin{equation*}
\X_{det}(T, X_0)=\left\{X:[0,T]\rightarrow \R^d \; \text{ absolutely continuous,}\; X_0\in\R^ d,\;\text{and}\; X_T=0 \right\},
\end{equation*}
and
\begin{equation*}
\cR_T^{\xi}=R_0+\int_0^T \big(X^\xi_t\big)^\top\sigma\;dB_t +\int_0^T b\cdot X^\xi_t \;dt-\int_0^T f(-\dot{\xi}_t)\;dt,
\label{rp}
\end{equation*}
denotes the revenues over the time interval $[0,T]$ associated to the process $X$.  Here $R_0\in\R$, $B$ is a standard $m$-dimensional Brownian motion starting in $0$ with drift $b\in\mathbb{R}^d$ (which is orthogonal to the kernel of the covariance matrix $\Sigma =\sigma\sigma ^\top$) and volatility matrix $\sigma=(\sigma^{ij})\in\R^{d\times m}$, and the  nonnegative, strictly convex function $f$ has superlinear growth and at most a polynomial growth of degree $p$, i.e.,  there exists $C>0$ such that $$ f(x)\leq C(1+|x|^p)\quad \text{ for all }x\in\R^d,$$and satisfies the two conditions $\lim_{|x|\longrightarrow \infty} \tfrac{f(x)}{|x|}=\infty$ and $f(0)=0$. Further, we will suppose that there exist two positive constants $A_i,i=1,2,$ such that
\begin{equation}
0< A_1\leq-\frac{u''(x)}{u'(x)}\leq A_2 \quad\text{for all } x\in\R.\label{apc}
\end{equation}
This inequality implies that we can assume w.l.o.g. that $0<A_1<1<A_2$,
which gives us the following estimates
 \begin{equation}\exp(-A_1 x)\leq u'(x)\leq \exp(-A_2x)+1\quad \text{ for } x\in\R.\label{ieu'}\end{equation}
 and
\begin{equation}
u_1(x):=\frac1{A_1}-\exp(-A_1x)\geq u(x)\geq -\exp(-A_2x)=: u_2(x).\label{ubd1}
\end{equation}
We refer to \citet{LM15} and \citet{LM15n} for more precisions and for the following results:
\begin{Theo}\label{eos}
Let  $\left(T,X_0,R_0\right)\in\;]0,\infty[\times\R^d\times\R$, then there exists a unique optimal  strategy  $\xi^*\in\dot{\X}^1(T,X_0)$ for the maximization problem \eqref{omp}, which satisfies 
 \begin{equation}
V(T,X_0,R_0)=\sup_{\xi\in\dot{\X}^1(T, X_0)}\E[u(\cR_{T}^{\xi})]=\E\Big[u\big(\cR^{\xi^*}_T\big)\Big],\label{omp1}
\end{equation}
\end{Theo}
\begin{Theo}\label{v_r}
The value function is concave and continuously partially differentiable in its third argument $R$, and we have the formula 
\[
V_r(T,X,R)=\E\big[u'\big(\cR_T^{\xi^*}\big)\big],
\]
 where $\xi^*$ is the optimal strategy associated to  $V(T,X,R)$.
\end{Theo}
The following result requires the notion of a comparison principle; see the Definition \ref{adv} below for a precise formulation.
\begin{Theo}\label{cv}
The value function $V$ fulfils a comparison principle and is thus the unique viscosity solution of the following HJB-equation with singularity
 \begin{align}
-V_t  +\frac{X^\top \Sigma X}{2}V_{rr}&+b\cdot X V_r+\sup_{\eta\in \R^d}\Big(\eta ^{\top}\nabla_x V-f(\eta)V_r\Big)T,X,R)=0,
\label{hjb}\\
V(0,X,R)&= \lim_{T\downarrow 0}V(T,X,R)=\begin{cases}
                                             u(R),& \text{if}\; X=0\\
                                     -\infty,& \text{otherwise}.\\
                                    \end{cases}\label{hjbic}
           \end{align}

\end{Theo}

\section{Auxiliary HJB equation, vanishing singularity and comparison result}

In this section we consider the following HJB equation:
\begin{align}
W_t -b\cdot X\,W_r - \frac{X^\top\Sigma X}2 \big(W_{rr}&+(W_r)^2\big)+\sup_{\xi\in \R^d}(\xi\cdot\nabla_xW+f(-\xi)W_r)=0\label{hjba}\\
W(0,X,R)= \lim_{T\downarrow 0}W(T,X,R)&=\begin{cases}
                                             \log (B-u(R)),& \text{if}\; X=0,\\
                                     \infty,& \text{otherwise},\\
                                    \end{cases}\label{hjbia}
 \end{align}
where $u$ denotes our utility function and $B\geq0$ is such that $B-u>0$ on $\R$ (such a $B$ exists, since the  utility function considered is bounded from above). We wish to show the equivalency between both preceding equations the viscosity sense. To this end, we recall briefly the definitions of viscosity sub- and supersolutions for continuous solutions.
Consider a nonlinear second-order degenerate partial differential equation
\\
\begin{equation}
F(T-t,x,r,v(T-t,x,r),v_t(t,x,r),\nabla_x v(t,x,r),v_r(t,x,r),v_{rr}(t,x,r))=0,\label{ape}
\end{equation}
\\
where $F$ is a continuous function on $]0,T]\times\R^{d}\times\R\times\R\times\R\times\R^d\times\R\times\R$ taking values in $\R$, with a fixed $T>0$ and $(t,x,r)\in\;]0,T]\times\R^{d}\times\R$ with the following assumption:
For all  $(t,x,r,q,p,s,m)\in\;]0,T]\times\R^{d}\times\R\times\R\times\R\times\R^d\times\R$ and $a,b\in\R$, we assume
$$
F(T-t,x,r,q,p,s,m,a)\leq F(T-t,x,r,q,p,s,m,b)\text{ if } a\geq b.
$$
\begin{Def}\label{adv}
Let $v: \;]0,T]\times\R^{d}\times\R\longrightarrow \R$  be a continuous function.
 \begin{enumerate}
\item We say that $v$ is a \emph{viscosity subsolution} of \eqref{ape} if for every $\varphi\in \C^{1,1,2}(]0,T]\times\R^{d}\times\R )$ and every $(t^*,x^*,r^*)\in [0,T[\times\R^{d}\times\R$, when
$v-\varphi$ attains a local maximum at $(T-t^*,x^*,r^*)\in\;]0,T]\times\R^{d}\times\R$, we have
\begin{equation*} 
                   F(.,v,\varphi_t,\nabla_x \varphi,\varphi_r,\varphi_{rr})(T-t^*,x^*,r^*)\leq0.
\label{asub}
\end{equation*}
\item We say that $v$ is a \emph{viscosity supersolution} of \eqref{ape} if for every $\varphi\in \C^{1,1,2}(]0,T]\times\R^{d}\times\R)$ and every $(t^*,x^*,r^*)\in [0,T[\times\R^{d}\times\R$, when
$v-\varphi$ attains a local minimum at $(T-t^*,x^*,r^*)\in\;]0,T]\times\R^{d}\times\R$, we have
\begin{equation*} 
                   F(.,v,\varphi_t,\nabla_x \varphi,\varphi_r,\varphi_{rr})(T-t^*,x^*,r^*)\geq0.
\label{asup}
\end{equation*}
\item We say that $v$ is  a \emph{viscosity solution} of the equation \eqref{ape} if $v$ is a viscosity subsolution and supersolution. 
\item We say that \eqref{ape} has a comparison result,  if for any subsolution $U$ and any supersolution $U$ satisfying the boundary condition
\begin{align*}
\limsup_{t\rightarrow 0} \big(U(t,x,r)-V(t,x,r)\big)&\leq 0, \quad \text{ for fixed } x,r\in\R^d\times\R.
\end{align*}
Then $U\leq V$ on $]0,T]\times\R^d\times\R$.
\end{enumerate}
\end{Def}  

\begin{Prop}\label{vef}
$U$ is a viscosity subsolution (resp.,~$V$ is a viscosity supersolution) of \eqref{hjb} if and only if $\log(B-U)$ is a viscosity supersolution (resp., $\log(B-V)$ is a viscosity subsolution) of \eqref{hjba}.
\end{Prop}
\begin{proof}
 We prove the following equivalence: $U$ is a viscosity subsolution of \eqref{hjb} if and only if $\log (B-U)$ is a viscosity supersolution of \eqref{hjba}. The other equivalence, i.e., $V$ is a viscosity supersolution of \eqref{hjb} if and only if $\log(B-V)$ is a viscosity subsolution of \eqref{hjba}, can be treated similarly.
 
  To this end, take $U$ a viscosity subsolution of \eqref{hjb}, $\varphi\in \C^{1,1,2}(]0,T]\times\R^{d}\times\R )$ and $(T-t^*,x^*,r^*)$ such that $(T-t^*,x^*,r^*)$ is a local minimizer of $\log(B-U)-\varphi$. We can w.l.o.g. suppose that $(\log(B-U)-\varphi)(T-t^*,x^*,r^*)=0$. Hence, we have that $B-U=\exp(\varphi)$ at $(T-t^*,x^*,r^*)$, and therefore
 it follows that $(T-t^*,x^*,r^*)$ is a local maximizer of $U-B+\exp(\varphi)$ (and also of $U+\exp(\varphi))$. We compute now the following derivatives of $\psi:=-\exp(\varphi)$ at $(T-t,x,r)$:
\begin{align*}
\psi_t&=-\varphi_t\psi,&\psi_r&=\varphi_r\psi,\\
\psi_{rr}&=(\varphi_{rr}+(\varphi_r)^2)\psi,&\nabla_x\psi&=\nabla_x\varphi\psi.
\end{align*}
Since $U$ is viscosity subsolution of \eqref{hjb}, we can write: 
\begin{IEEEeqnarray*}{rCl}
\IEEEeqnarraymulticol{3}{l}{\Big(-(\psi)_t +\frac{X^\top\Sigma X}2 \psi_{rr}+b\cdot X\,(\psi)_r +\sup_{\xi\in \R^d}(\xi\cdot\nabla_x\psi-f(\xi)\psi_r)\Big)(T-t^*,x^*,r^*)}\\
&=& \psi\Big(-\varphi_t+b\cdot X\,(\varphi)_r+ \frac{X^\top\Sigma X}2 \varphi_{rr}+ \frac{X^\top\Sigma X}2(\varphi_r)^2\\
&&-\>\sup_{\xi\in \R^d}\big(-\xi\cdot\nabla_x\varphi+f(\xi)\varphi_r)\Big)(T-t^*,x^*,r^*)\\
&\geq& 0.
\end{IEEEeqnarray*}
Hence, we get that
\[
\varphi_t-b\cdot x\,(\varphi)_r- \frac{X^\top\Sigma X}2 \big(\varphi_{rr} +(\varphi_r)^2\big)+\sup_{\xi\in \R^d}\big(\xi\cdot\nabla_x\varphi+f(-\xi)\varphi_r)\geq 0,
\]
at $(T-t^*,x^*,r^*)$, which proves the one direction.

The converse direction can be proved in a very similar way.
\end{proof}
We show now that a comparison principle also holds for \eqref{hjba}.
\begin{Prop}\label{p}
Let $W$ (resp., $\widetilde{W}$) be a continuous viscosity subsolution (resp., continuous viscosity supersolution) of \eqref{hjba}, defined on $]0,T]\times\R^d\times\R$, which satisfies the growth conditions
 \begin{equation}
\log(B-V_2(t,x,r))\geq v(t,x,r)\geq \log(B-V_1(t,x,r)), \text{ for all }(t,x,r)\in\;]0,T]\times\R^d\times\R,\label{gc}
\end{equation}
where $v$ can be chosen to be $W$ or $\widetilde{W}$. Further, we suppose that $W$ and $\widetilde{W}$  satisfy the boundary conditions
\begin{align}
 \limsup_{t\rightarrow 0} W(t,x,r)-\widetilde{W}(t,x,r)&\leq 0, \quad \text{ for fixed } x,r\in\R^d\times\R.\label{u-v0}
\end{align}
Then $W\leq \widetilde{W}$ on $]0,T]\times\R^d\times\R$.
\end{Prop}
\begin{proof}
We write $\widetilde{W}=\log(B-\widetilde{U})$ and $W=\log(B-U)$. Then, by applying Proposition \ref{vef} we have that $U$ is a supersolution (resp., $\widetilde{U}$ is a subsolution) of \eqref{hjb}, and satisfies
\[
\limsup_{T\downarrow 0} (U-\widetilde{U})(T,X_0,R_0)\geq 0.
\]
Thus, as  \eqref{hjb} fulfils a comparison principle, we have that $U\geq \widetilde{U}$, which implies that $W\leq\widetilde{W}$ on $]0,T]\times\R^d\times\R$.
\end{proof}
The preceding important result permit us to relax the exponential growth condition imposed on the value function. By using an affine transform of the preceding HJB equation, with an adequate function, we will also be able to remove the singularity in the initial condition.  To this end, we  first need to prove the following fundamental proposition.
\begin{Prop}\label{iws} 
Define $\tilde{u}(T,X_0,R_0):=\log (B-u(R_0-Tf(-X_0/T))$, and  let $V$ denote the value function of the maximization problem \eqref{omp} with initial condition \eqref{hjbic}. Then
 $\tilde{u}\in\C^{1,1,2}(]0,T]\times\R^{d}\times\R )$ and verifies
\begin{equation}
\lim_{T\downarrow0}\log(B-V(T,X_0,R_0))-\tilde{u}(T,X_0,R_0)=0,
\end{equation}
locally uniformly in $ (X_0,R_0)$.
\end{Prop}
\begin{proof}
It is sufficient to prove that, for $X_0\neq0$, it holds
\[
\lim_{T\downarrow0}\frac{V(T,X_0,R_0)}{u(R_0-Tf(-X_0/T))}= 1. 
\]
 Toward this end, consider first the linear strategy
 $\zeta:=X_0/T\in\dot{\X}^1(T,X_0)$. We want to show that
\begin{equation}
\lim_{T\downarrow0}\frac{\E\big[u(\cR^\zeta_T)\big]}{u(R_0-Tf(-X_0/T))}=1,\label{lvu}
\end{equation}
where
\[
\cR^\zeta_T=R_0+X_0\int^T_0(1-t/T)\sigma^\top\,dB_t+\frac{T}2b\cdot X_0-Tf(-X_0/T).
\]
But we have
\begin{IEEEeqnarray*}{rCl}
 \E\big[u(\cR^\zeta_T)\big] &=& \E\big[u(\cR^\zeta_T+A_2/2 \langle \cR_.^{\zeta}\rangle_T-A_2/2 \langle \cR_.^{\zeta}\rangle_T)\big)\big]\\
 &=&\E\bigg[u \bigg(R_0+\frac{T}2b\cdot X_0-Tf(X_0/T)-\frac{A_2}2\int_0^T(X^{\zeta}_t)^\top \Sigma X^{\zeta}_t \,\d\bigg)\bigg]\\
&&+\>\E\bigg[\int^T_0 u'(\cR_t^{\zeta})\,X^{\zeta}_t \sigma\,dB_t+\frac{A_2}2\int^T_0 u'(\cR_t^{\zeta})(X^{\zeta}_t)^\top \Sigma
 X^{\zeta}_t \,\d\\
 &&+\>\frac1{2}\int^T_0 u''(\cR_t^{\zeta})\,d
 \langle \cR_{ \cdot}^{\zeta}\rangle_t \bigg]\\
    &\geq&u \bigg(R_0+\frac{T}2b\cdot X_0-Tf(X_0/T)-2A_2|X_0|^2T|\Sigma|\bigg)\\    
 &&+\>\E\bigg[\frac{A_2}2\int^T_0 u'(\cR_{t}^{\zeta})(X^{\zeta}_t)^\top \Sigma
  X^{\zeta}_t \,\d\\
  &&-\>\frac1{2}\int^T_0 u'(\cR_{t}^{\zeta})\frac{-u''(\cR_t^{\zeta})}{u'(\cR_t^{\zeta})}\,
  (X^{\zeta}_t)^\top \Sigma
  X^{\zeta}_t \,\d\bigg]\\
   &\geq&u \bigg(R_0+\frac{T}2b\cdot X_0-Tf(X_0/T)-2A_2|X_0|^2T|\Sigma|\bigg)\\
   &&+\>\E\bigg[\frac{A_2}2\int^T_0 u'(\cR_{t}^{\zeta})(X^{\zeta}_t)^\top \Sigma
 X^{\zeta}_t \,\d-\frac{A_2}2\int^T_0 u'(\cR_{t}^{\zeta})(X^{\zeta}_t)^\top \Sigma
 X^{\zeta}_t \,\d\bigg]\\
&=&u \bigg(R_0+\frac{T}2b\cdot X_0-Tf(X_0/T)-2A_2|X_0|^2T|\Sigma|\bigg).
  \end{IEEEeqnarray*}
 And this implies that
 \begin{equation}
\liminf_{T\downarrow0}\frac{u(R_0-Tf(-X_0/T))}{\E\big[u(\cR^\zeta_T)\big]}\geq1.
\end{equation}
 Let now $\xi^*$ be the optimal strategy associated to $V(T,X_0,R_0)$. Observe that applying Jensen's inequality  to the convex function $f$ and the concave function $u$ yields the inequality
 \[
\E\big[u\big(\cR^{\xi^*}_T\big)\big]\leq u\bigg(\E\bigg[R_0+\int^T_0 X^{\xi^*}_t\cdot b\,dt-Tf(-X_0/T)\bigg]\bigg).
\]
As $\xi^*\in\dot{\X}^1(T,X_0)$, we can find an $M>0$ such that $\E[\int^T_0 X^{\xi^*}_t\cdot b\,dt]\leq |b|MT$. And therefore we have
\begin{equation}
\E\big[u\big(\cR^{\xi^*}_T\big)\big]\leq u\Big(R_0+|b|MT-Tf(-X_0/T)\Big).\label{vfi}
\end{equation}
Using the fact that, for $T$ close enough to $0$, both $V(T,X_0,R_0)$ and $u(R_0-Tf(-X_0/T))$ are negative, we finally get 
\begin{equation}
\liminf_{T\downarrow0}\frac{V(T,X_0,R_0)}{u(R_0-Tf(-X_0/T))}\geq \liminf_{T\downarrow0}\frac{u(R_0+|b|MT-Tf(-X_0/T))}{u(R_0-Tf(-X_0/T))}= 1.\label{u}
\end{equation}
Since $\xi^*$ is optimal (and hence $V(T,X_0,R_0)\geq\E\big[u(\cR^\zeta_T)\big]$), we also have 
\begin{align*}
1\geq&\limsup_{T\downarrow0}\frac{V(T,X_0,R_0)}{\E\big[u(\cR^\zeta_T)\big]}\\
=&\limsup_{T\downarrow0}\frac{V(T,X_0,R_0)u(R_0-Tf(-X_0/T))}{\E\big[u(\cR^\zeta_T)\big]u(R_0-Tf(-X_0/T))}\\
=&\limsup_{T\downarrow0}\frac{V(T,X_0,R_0)}{u(R_0-Tf(-X_0/T))}\cdot\liminf_{T\downarrow0}\frac{u(R_0-Tf(-X_0/T))}{\E\big[u(\cR^\zeta_T)\big]}\\
=&\limsup_{T\downarrow0}\frac{V(T,X_0,R_0)}{u(R_0-Tf(-X_0/T))}.
\end{align*}
Combining the preceding inequality with \eqref{u} concludes the proof
\end{proof}
  \begin{rem}\label{bcs}
 The preceding proof remains unchanged if we send $|R_0|$  to infinity (instead of sending $T$ to $0$),  other parameters being fixed. Thus, we have that
 \begin{equation*}
\lim_{R_0\rightarrow\pm\infty}\log(B-V(T,X_0,R_0))-\log(B-u(R_0-Tf(X_0/T))=0.                              
\end{equation*}
This will later enable us  to set  $\log(B-u(R_0-Tf(X_0/T))$ as a boundary condition, when taking $|R_0|$ large enough in our scheme (since we will work with a finite grid in the numerical examples);
however,  in general,
\begin{equation*}
\lim_{|X_0|\rightarrow \infty}\log(B-V(T,X_0,R_0))-\log(B-u(R_0-Tf(X_0/T))\neq0,                              
\end{equation*}
for $T\neq0$.
\xqed\diamondsuit
 \end{rem}

For $\tilde{u}$ as in the preceding proposition, we consider now the following auxiliary equation with zero as initial condition for \eqref{hjba}:
\begin{align}
(W+\tilde{u})_t &-b\cdot X\,(W+\tilde{u})_r- \frac{X^\top\Sigma X}2\big( (W+\tilde{u})_{rr}+((W+\tilde{u})_r)^2\big)\notag\\
&+\sup_{\xi\in \R^d}\big(\xi\cdot\nabla_x(W+\tilde{u})+f(-\xi)(W+\tilde{u})_r)=0,&\label{hab}\\
& \lim_{T\downarrow 0}W(T,X,R)\label{habi}=0.& \end{align}
\begin{rem}
 Note that we can rewrite \eqref{hab} in the following way:
\begin{align*}
0&=(W+\tilde{u})_t -b\cdot X\,(W+\tilde{u})_r - \frac{X^\top\Sigma X}2 \big((W+\tilde{u})_{rr}+ ((W+\tilde{u})_r)^2\big)\\
&-(W+\tilde{u})_r f^*\bigg(-\frac{\nabla_x(W+\tilde{u})}{(W+\tilde{u})_r}\bigg),
\end{align*}
where $f^*$ denotes the Fenchel-Legendre transformation of $f$.
\xqed{\diamondsuit}
\end{rem}
 The next proposition states that the notion of viscosity solutions of \eqref{hjba} and viscosity solutions of  \eqref{hab} is equivalent, and moreover,  a comparison result holds. 
\begin{Prop}\label{cpb}
$W$ is a viscosity subsolution (resp., supersolution) of \eqref{hjba} with initial condition \eqref{hjbia} if and only if $W-\tilde{u}$ is a viscosity subsolution (resp., supersolution) of \eqref{hab} with initial condition \eqref{habi}. Moreover, a comparison principle holds for \eqref{hab}.
\end{Prop}
\begin{proof}
This is a straightforward application of Proposition \ref{p} and the definition of viscosity solutions: we have that $\varphi$ is a test function for $W$, when applied to \eqref{hjba}, if and only if $\varphi-\tilde{u}$ is a test function for $W-\tilde{u}$, when applied to \eqref{hab}.
\end{proof}
\section{Numerical schemes and convergence results}
In this section, our goal is to prove a convergence result, similar to the one derived in \citet{BS91}. However, we will have to relax their conditions in order to ensure that finite difference schemes applied to our numerical examples will converge, locally uniformly, to the unique viscosity solution of \eqref{hab}. Let us now introduce the definition of a numerical scheme, in our setting.
\subsection{Barles-Souganidis convergence result}
\begin{Def}\label{sd}
A numerical scheme for \eqref{hab} with initial condition \eqref{habi} is an equation of the following form:
\begin{align}
S(h,t,x,r,w_h(t,x,r),[w_h]_{t,x,r})&=0, \quad \text{ for } (t,x,r)\in \G_h\backslash \{t=0\},\label{ds}\\
w_h(0,x,r)&=0,\quad \text{ in } \G_h\cap\{t=0\},\label{dsic}
\end{align}
where $S$ is locally bounded, $h:=\max(|\Delta t|, |\Delta x|, |\Delta r|)$  denotes the size of the mesh, and $$\G_h:=\Delta t\cdot\{0,1,\dots,n_T\}\times\Delta x\cdot \Z^d\times \Delta r\cdot\Z.$$ The quantity $w_h$ represents the approximation of $w$, and  $[w_h]_{t,x,r}$ stands for the value of $w_h$ close to $(t,x,r)$.
\end{Def}
 In order to have an analogous result to the Barles-Souganidis convergence theorem that can be applied to our numerical schemes, we need to slightly modify the three conditions required  in \citet{BS91}.
\begin{Def}
A numerical scheme $S$ is said to be
\begin{itemize}
\item \emph{locally $\delta$-monotone} if there exists $\delta>0$ such that whenever $|w-v|\leq\delta$: if $w\geq v$ on an open bounded set $O\subset\;]0,T]\times \R^d\times\R$, then
 $$S(h,t,x,r,z,w)\leq S(h,t,x,r,z,v),$$
for all $h>0,\; (t,x,r)\in O$ and $z\in\;]-S_O,S_O[$, where $S_O:=\sup_{y\in O}|w(y)|+1$. Here, $w\geq v$ is to be understood componentwise.
\item \emph{consistent} if, for every $\varphi\in \C^{1,1,2}(]0,T]\times\R^{d}\times\R )$ and every $(t,x,r)\in [0,T[\times\R^{d}\times\R$, we have
\begin{eqnarray*}
\lefteqn{S(h,t,x,r,\varphi(t,x,r), [\varphi+m ]_{t,x,r})}\\
&&\underset{\substack{
  m\to0\\
   h\to0
  }}{\longrightarrow}\big((\varphi+\widetilde{u})_t -\frac{x^\top\Sigma x}2(\varphi+\tilde{u})_r^2-\inf_{\xi\in \R^d}\widetilde{\cL}^{\xi}(\varphi+\tilde{u})\big)(t,x,r),
\end{eqnarray*}
locally uniformly in $(t,x,r)$, with
\begin{eqnarray*}
\widetilde{\cL}^{\xi}(\varphi+\tilde{u})\big)(t,x,r)&=&\frac{x^\top\Sigma x}2 (\varphi+\tilde{u})_{rr}+b\cdot x\,(\varphi+\tilde{u})_r \\
&& -\big(\xi\cdot\nabla_x(\varphi+\tilde{u})+f(-\xi)(\varphi+\tilde{u})_r\big)(t,x,r).
\end{eqnarray*}
\item \emph{(locally) stable} if there exists $\delta>0$ such that, for every $\delta>h>0$ and every open bounded set $O\subset\;]0,T[\times\R^{d}\times\R$, there is a locally bounded solution $w_h$ of \eqref{ds} satisfying
\[
\sup_{h>0} |w_h|\leq C_O \text{ on }\; O,
\]
where $C_O$ is a constant depending only on $O$.
\end{itemize}
\end{Def}
\begin{rem}
\begin{enumerate}
\item In the preceding definition, the monotonicity property as defined in \citet{BS91} (i.e., monotonicity of the scheme without requiring an additional control of $|w-v|$)  can be replaced by our $\delta$-monotonicity, as mentioned by \citet{Ta11}.
\item The local stability is equivalent to the one used by Barles and Souganidis, due to the local property of the viscosity solution.
\item Since the viscosity solution of \eqref{hab} is continuous and has a partial derivative in its third variable (Theorem \ref{v_r}), the approximation $w_h$ can be chosen among the same class of functions. Moreover, as this partial derivative has locally a strictly negative upper bound, we can suppose that the analogous boundedness property also holds  for $w_h$.
\item
As for the comparison principle, the monotonicity property is crucial, and without this assumption the scheme may fail to converge to the unique viscosity solution, as it can be seen in, e.g., \citet{PF03} or \citet{O06nn}. This property is in practice the most difficult one to prove, due to the nonlinearity of our HJB equation, as we will see in the next section.
\end{enumerate}
\xqed{\diamondsuit}
\end{rem}
We can now state and show the fundamental theorem of this chapter. 
\begin{Theo}\label{bsc}
Suppose that the numerical scheme $S$ is $\delta$-monotone, consistent, and locally stable. Then, the solution $w_h$ of \eqref{ds} converges, locally uniformly on the set $]0,T]\times\R^d\times\R$, to the unique continuous viscosity solution of \eqref{hab}.
\end{Theo}
\begin{proof}
Take $(t,x,r)\in\;]0,T]\times\R^d\times\R$ and let us define $w^*,w_*$ as follows:
\begin{align}
w^*(t,x,r):= \limsup _{\begin{subarray}{c} h \to 0 \\ (t',x',r') \to (t,x,r) \end{subarray}} w_h(t',x',r') \quad\text{and}\quad 
w_*(t,x,r):= \liminf _{\begin{subarray}{c}  h \to 0 \\ (t',x',r') \to (t,x,r)\end{subarray}} w_h(t',x',r').\label{wssd}
 \end{align}
These quantities are known as the classical half-relaxed limits and, due to the local stability assumption, $w^*$ and $w_*$ are well-defined. Suppose first that $w^*$ and $w_*$ are viscosity sub- and supersolution of \eqref{hab},  respectively, and verify
 \begin{equation}
 \limsup_{t\rightarrow 0} w^*(t,x,r)-w_*(t,x,r)\leq 0, \label{vssl}
\end{equation}
 whence we can infer  (Proposition \ref{cpb}) that $w^*\leq w_*$. Since we also have that $w^*\geq w_*$, by definition \eqref{wssd}, we then obtain that $w^*=w_*$ is the unique viscosity solution of \eqref{hab}. Hence, it is sufficient to show that $w^*$ and $w_*$ are viscosity sub- and supersolution of \eqref{hab}, respectively. 

We start by proving that $w^*$ is a subsolution. To this end, take $\varphi\in \C^{1,1,2}(]0,T]\times\R^{d}\times\R )$ such that $w^*-\varphi$ attains its maximum on a bounded open set $O,$ at some $(T-t^*,x^*,r^*)\in\;]0,T]\times\R^d\times\R$. As already argued, by translating $\varphi$ if necessary, we can w.l.o.g suppose that 
\begin{equation}
(w^*-\varphi)(T-t^*,x^*,r^*)=0,\label{0}
\end{equation}
and that this maximum can be taken as strict.
Due to the definition of $w^*$, we can find sequences $h_n$ and $(T-t^{h_n},x^{h_n},r^{h_n})\in O$, such that $h_n\downarrow0,\;(T-t^{h_n},x^{h_n},r^{h_n})\to(T-t^*,x^*,r^*)
$ and 
\begin{equation}
(w_{h_n}-\varphi)(T-t^{h_n},x^{h_n},r^{h_n})-h_n\uparrow (w^*-\varphi)(T-t^*,x^*,r^*). \label{52}
\end{equation}
Hence, by taking a subsequence if necessary, we have that $(w_{h_n}-\varphi)$ also attains its maximum on $O$,  at some $(T-t^{h_n},x^{h_n},r^{h_n})$, i.e.,
\begin{equation}
w_{h_n}(T-t,x,r)\leq \varphi (T-t,x,r)+ (w_{h_n}-\varphi)(T-t^{h_n},x^{h_n}).
\end{equation}
Indeed, for $(T-t,x,r)\in O$ we can write on one hand
\begin{align*}
 (w^*-\varphi)(T-t^*,x^*,r^*)&>(w^*-\varphi)(T-t,x,r)\\
 &= \limsup _{\begin{subarray}{c} h \to 0 \\ (t',x',r') \to (t,x,r) \end{subarray}} w_h(t',x',r')-\varphi(T-t,x,r)\\
 &\geq w_{h_n}(T-t,x,r)-\varphi(T-t,x,r)-h_n,
\end{align*}
due to \eqref{52}, for all $n$ taken large enough. On the other hand, we can also write (by using again \eqref{52})
\begin{align*}
 (w^*-\varphi)(T-t^*,x^*,r^*)&\geq(w_{h_n}-\varphi)(T-t^{h_n},x^{h_n},r^{h_n})-h_n\\
 &>(w^*-\varphi)(T-t,x,r),
\end{align*}
for some $n\in\N$ taken large enough. Further, using \eqref{0} and the continuity of both $w_{h_n}$ (see preceding remark) and $\varphi$ (taking $O$ smaller if necessary), we have that $|w_{h_n}-(\varphi+m_n)|\leq \delta$ on $O$,  where $$m_n:=(w_{h_n}-\varphi)(T-t^{h_n},x^{h_n},r^{h_n}).$$ Applying the $\delta$-monotonicity property of the scheme to $\varphi+m_n$ and using the fact that $w_{h_n}$ is a solution of \eqref{ds} yields:
\[
S(h^n,T-t^{h_n},x^{h_n},r^{h_n},\varphi(T-t^{h_n},x^{h_n},r^{h_n}), [\varphi+m_n]_{t,x,r})\leq 0.
\]
Utilizing moreover the fact that, as $h^n\to0,$ it holds that $m_n\to (w^*-\varphi)(T-t^*,x^*,r^*)$ and the consistency of the scheme, we  infer that
\[
\bigg((\varphi+\widetilde{u})_t -\frac{(x^*)^\top\Sigma x^*}2(\varphi+\tilde{u})_r^2-\inf_{\xi\in \R^d}\widetilde{\cL}^{\xi}(\varphi+\tilde{u})\bigg)(T-t^*,x^*,r^*)\leq0,
\]
which proves that $w^*$ is a subsolution of \eqref{hab}. In the same manner, we can prove that $w_*$ is a viscosity supersolution. Since we also have that \eqref{vssl} is verified, due to \eqref{dsic}, our theorem is established.
\end{proof}
In the next step, we are going to apply the preceding results to construct converging numerical schemes. In particular, we will deal with two types  of schemes: explicit and implicit schemes. While the first one is easy to apply, it  also requires us to take a very small time step, compared to the other step parameters, whereas the  second one does not have any restriction at all with the time step. It is  however essentially more difficult to numerically apply the implicit scheme. For the sake of simplicity, we will restrict ourselves to the three-dimensional case (i.e.,  $d=1$).
\subsection{Construction of a converging explicit scheme}
Establishing the local $\delta$-monotonicity property of a scheme can be very challenging, in general, even in linear cases. This is mostly the case for explicit schemes for the equation \eqref{hab}, which shows that the Barles-Souganidis convergence result is quite difficult to apply, here. Before we construct such a scheme, we first need to make the following assumptions: 
 \begin{Ass}\label{lo}
 We restrict ourselves to the situation where the solution of \eqref{hab} is locally Lipschitz-continuous in the second parameter $x$, i.e.,  for every bounded set $O\subset\;]0,T[\times\R^{d}\times\R$, there exists $L_O>0$ such that, for every $(t,x,r)\in O$ we have
\[
\limsup_{h\rightarrow0}\bigg|\frac{W(t,x+h,r)-W(t,x,r)}h\bigg|\leq L_O.
\]
We suppose that this is also the case for the partial derivative $W_r$, i.e.,   for every bounded set $O\subset\;]0,T[\times\R^{d}\times\R$, there exists $\overline{K}'_O>0$ such that, for every $(t,x,r)\in O$  
\[
\limsup_{h\rightarrow0}\bigg|\frac{W_r(t,x,r+h)-W_r(t,x,r)}h\bigg|\leq \overline{K}'_O.
\]
\end{Ass}
\begin{rem}
Since $W_r$ is continuous, we automatically have  that $W$ is locally Lipschitz-continuous in its third parameter, $r$. Hence, there exists $\overline{K}_O>0$  such that, for every $(t,x,r)\in~O$ 
\[
\limsup_{h\rightarrow0}\bigg|\frac{W(t,x,r+h)-W(t,x,r)}h\bigg|\leq \overline{K}_O.
\]
\xqed\diamondsuit
\end{rem}
Even by considering a simple standard explicit scheme with no drift, it seems to be difficult, even impossible, to establish a condition  on $\Delta t, \Delta x, \Delta r$ such that such scheme fulfils a (local) monotonicity property. 
 We thus need to modify our preceding scheme by taking into account  the following facts:
\begin{enumerate}
\item Starting from the upwind schemes for $\tilde{w}_x$, $$ \tilde{w}_x=\tilde{w}_i-\tilde{w}_{i-1}\quad \text{and} \quad-\tilde{w}_x=\tilde{w}_i-\tilde{w}_{i+1},$$ and using $|x|=\max(x,-x)$,  $x^2=|x|^2,$ we can obtain the following scheme for $\tilde{w}_x^2$:
\begin{align*}
 \tilde{w}_x^2&=\frac1{\Delta x} \max (\tilde{w}_i-\tilde{w}_{i-1},\tilde{w}_i-\tilde{w}_{i+1},0)^2,
\end{align*} 
 in which we omit the index of the non-concerned terms. 
\item Since $W_r$ is continuous, we can approximate it by either $(\tilde{w}_{k}-\tilde{w}_{k-1})/\Delta r$ or $(\tilde{w}_{k+1}-\tilde{w}_{k})/\Delta r$. Since $V_r$ is strictly positive on $]0,T]\times\R^d\times\R$, we have that $W_r=\log(B-V)_r$ is strictly negative and hence, on every bounded set $O\subset\;]0,T]\times\R^d\times\R$ there exists  $K_O>0$ such that $W_r<-K_O$ on $O$. Thus, we can suppose that 
\begin{equation}
\max\big\{(\tilde{w}_{k+1}-\tilde{w}_{k})/\Delta r,(\tilde{w}_{k}-\tilde{w}_{k-1})/\Delta r\big\} <-K_O.\label{ko}
\end{equation}
\end{enumerate}
These considerations show that we may have to consider the following explicit scheme:
\begin{IEEEeqnarray*}{rCl}
\IEEEeqnarraymulticol{3}{l}{S(h,\Delta t,\Delta x,\Delta r,\tilde{w}_{i,k}^{n+1},[\tilde{w}_{i+1,k}^n, \tilde{w}^n_{i-1,k}, \tilde{w}_{i,k+1}^n, \tilde{w}^n_{i,k-1}, \tilde{w}_{i,k}^{n}])}\\
&=&\frac{\tilde{w}_{i,k}^{n+1}-\tilde{w}_{i,k}^n}{\Delta t}+\frac12\bigg(\frac{i\Delta x\sigma}{\Delta r}\bigg)^2 \big(\tilde{w}_{i,k}^n -\tilde{w}_{i,k-1}^n +\tilde{w}_{i,k}^n-\tilde{w}_{i,k+1}^n -\>\big(\tilde{w}_{i,k}^n -\tilde{w}_{i,k+1}^n \big)^2\big)\\
&&-\>\frac{\Delta r}{4\lambda(\Delta x)^2}\cdot\frac{\max\big(\tilde{w}_{i,k}^n -\tilde{w}_{i-1,k}^n,\tilde{w}_{i,k}^n -\tilde{w}_{i+1,k}^n,0 \big)^2 }{\tilde{w}_{i,k}^n -\tilde{w}_{i,k-1}^n }.                        
\end{IEEEeqnarray*}
Its stencil is represented below:\\
\begin{tikzpicture}
\matrix (m) [matrix of math nodes, row sep=3em,
column sep=3em]{
& \tilde{w}^{n}_{i,k+1}& &  \\
\tilde{w}^{n}_{i-1,k} &\tilde{w}^{n}_{i,k} & \tilde{w}^{n}_{i+1,k} & \\
&\tilde{w}^{n}_{i,k-1}  & \tilde{w}^{n+1}_{i,k}&  \\};
 \path
(m-1-2) [densely dotted] edge (m-2-2)
(m-2-1) edge (m-2-2)
(m-2-2)edge (m-2-3)
(m-3-2)edge (m-2-2);
 \path[-stealth]
 (m-2-2) edge (m-3-3);
\end{tikzpicture}

In the following, we show that this scheme converges  to the unique viscosity solution of \eqref{hab}. We begin by  proving the local $\delta$-monotonicity of the scheme, where it is moreover shown that $\delta$ can be taken as $1/2$. To this end, take an open bounded set $O\subset\;]0,T]\times \R^d\times\R$. First, note that our scheme $S$  is unconditionally decreasing in $\tilde{w}_{i,k-1}$. It is also nonincreasing in $\tilde{w}_{i+1,k}^n$ and in $\tilde{w}^n_{i-1,k}$ (recall that $(\tilde{w}_{k}-\tilde{w}_{k-1})/\Delta r<0$). Further, $S$ is nonincreasing in $\tilde{w}^{n}_{i,k+1}$ for $|\tilde{w}^{n}_{i,k}-\tilde{w}^{n}_{i,k+1}|\leq1/2$, because the function $x-x^2$ is nondecreasing for $-1/2\leq x\leq 1/2$.\\
We now prove  that $S$ is nonincreasing in $\tilde{w}^{n}_{i,k}$. This is the most difficult part of proving the monotonicity property of $S$, and we will only give  a sufficient condition for it (CFL-type condition).\\
\underline{First case:}  $\max\big(\tilde{w}_{i,k}^n -\tilde{w}_{i-1,k}^n,\tilde{w}_{i,k}^n -\tilde{w}_{i+1,k}^n,0 \big)=0$.\\
Consider the function
$$ \psi_1: \tilde{w}^{n}_{i,k}\longmapsto -\tilde{w}^{n}_{i,k}+ \frac{\Delta t}2\bigg(\frac{i\Delta x\sigma}{\Delta r}\bigg)^2 \big(2\tilde{w}^{n}_{i,k}-(\tilde{w}^{n}_{i,k}-\tilde{w}^{n}_{i,k+1})^2\big),$$
whose derivative is given by 
$$ \psi'_1: \tilde{w}^{n}_{i,k}\longmapsto -1+ \frac{\Delta t}2\bigg(\frac{i\Delta x\sigma}{\Delta r}\bigg)^2 \big(2-2(\tilde{w}^{n}_{i,k}-\tilde{w}^{n}_{i,k+1})\big).$$
Then, for $|\tilde{w}^{n}_{i,k}-\tilde{w}^{n}_{i,k+1}|\leq1/2$ and $\displaystyle{3\Delta t/2(i\Delta x\sigma/\Delta r)^2}\leq1$, we have that $\psi'_1\leq0$, and  $S$ is hence nonincreasing in $\tilde{w}_{i,k}^n$.\\
\underline{Second case:}  $\max\big(\tilde{w}_{i,k}^n -\tilde{w}_{i-1,k}^n,\tilde{w}_{i,k}^n -\tilde{w}_{i+1,k}^n,0 \big)\neq 0$.\\
We can suppose w.l.o.g. that $\max\big(\tilde{w}_{i,k}^n -\tilde{w}_{i-1,k}^n,\tilde{w}_{i,k}^n -\tilde{w}_{i+1,k}^n,0 \big)=\tilde{w}_{i,k}^n -\tilde{w}_{i-1,k}^n$. Consider now the following function:
$$ \psi_2: \tilde{w}^{n}_{i,k}\longmapsto -\tilde{w}^{n}_{i,k}+ \frac{\Delta t}2\bigg(\frac{i\Delta x\sigma}{\Delta r}\bigg)^2 \big(2\tilde{w}^{n}_{i,k}-(\tilde{w}^{n}_{i,k}-\tilde{w}^{n}_{i,k+1})^2\big)
-\frac{\Delta r\Delta t}{4\lambda(\Delta x)^2}\frac{(\tilde{w}_{i,k}^n -\tilde{w}_{i-1,k}^n)^2 }{\tilde{w}_{i,k}^n -\tilde{w}_{i,k-1}^n},$$
whose derivative is given by
\begin{IEEEeqnarray*}{rCl} 
\IEEEeqnarraymulticol{3}{l}{\psi'_2: \tilde{w}^{n}_{i,k}\longmapsto -1+ \frac{\Delta t}2\bigg(\frac{i\Delta x\sigma}{\Delta r}\bigg)^2 \big(2-2(\tilde{w}^{n}_{i,k}-\tilde{w}^{n}_{i,k+1})\big)}\\
&&-\frac{\Delta r\Delta t}{4\lambda(\Delta x)^2}\frac{(\tilde{w}_{i,k}^n -\tilde{w}_{i-1,k}^n)(\tilde{w}_{i,k}^n -\tilde{w}_{i,k+1}^n+\tilde{w}_{i-1,k}^n-\tilde{w}_{i,k}^n+\tilde{w}_{i,k}^n-\tilde{w}_{i,k+1}^n) }{(\tilde{w}_{i,k}^n -\tilde{w}_{i,k-1}^n)^2}.
\end{IEEEeqnarray*}
As $W$ is known to be continuous, it is uniformly continuous on any bounded set $O\subset\;]0,T]\times\R\times\R$ (where $\overline{O}\subset \;]0,T]\times\R\times\R$) and thus, there exists $h>0$ such that  $|\tilde{w}^{m}_{j,l}-\tilde{w}^{m'}_{j',l'}|\leq1/2,$ for $|(m,j,l)-(m',j',l')|\leq h$. We denote by $X_O$ the maximum value of $|i\Delta x|$ on $O\cap\R$. Using the fact that  $|(\tilde{w}^{m}_{j,l}-\tilde{w}^{m}_{j+1,l})/\Delta r|\geq K_O$ on $O$ (due to \eqref{ko}), we infer 
\begin{align*}
\psi'_2( \tilde{w}^{n}_{i,k})&\leq -1+ \frac{3\Delta t}2\bigg(\frac{i\Delta x\sigma}{\Delta r}\bigg)^2 +\frac{\Delta r\Delta t}{4\lambda(\Delta x)^2}\frac{(1/2)(3/2) }{(\Delta r)^2K_O^2}\\
&= -1+ \frac{3\Delta t}{8\lambda}\bigg(\frac{(i\Delta x\sigma)(\Delta x)^2 +\Delta r}{(\Delta x)^2(\Delta r)^2K_O^2}\bigg)\\
&\leq 0,
\end{align*}
for 
\begin{equation}
\frac{3\Delta t}{8\lambda}\bigg(\frac{(X_O\sigma)^2(\Delta x)^2K_O^2+\Delta r}{(\Delta r)^2(\Delta x)^2 K_O^2}\bigg)\leq 1.\label{cfl1}
\end{equation}
The condition \eqref{cfl1} can be regarded as the Courant Friedrichs Lewy (CFL) condition for this explicit scheme.\\
It remains to prove the consistency and local stability of the scheme.  Classical computations using the Taylor expansion yield:
\begin{align*}
 \frac{\tilde{w}_{i,k+1}^{n}+\tilde{w}^{n}_{i,k-1}-2\tilde{w}_{i,k}^{n}}{(\Delta r)^2}&=\tilde{w}_{rr}(n\Delta t, i\Delta x, k\Delta r)\\
 &\quad+\>\frac1{12}\tilde{w}_{rrrr}(n\Delta t, i\Delta x, k\Delta r) (\Delta r)^2+ o(\Delta r)^2,\\
 \frac{\tilde{w}_{i,k+1}^{n}-\tilde{w}_{i,k}^{n} }{\Delta r}&=\tilde{w}_r(n\Delta t, i\Delta x, k\Delta r)+\frac12\tilde{w}_{rr}(n\Delta t, i\Delta x, k\Delta_r)\Delta r\\
 &\quad+\>o(\Delta r),\\
 \frac{\tilde{w}_{i,k}^{n}-\tilde{w}_{i,k-1}^{n} }{\Delta r}&=\tilde{w}_r(n\Delta t, i\Delta x, (k-1)\Delta r)\\
          &\quad+\>\frac12\tilde{w}_{rr}(n\Delta t, i\Delta x, (k-1)\Delta_r)\Delta r+o(\Delta r),\\
  \frac{\tilde{w}_{i+1,k}^{n}-\tilde{w}_{i,k}^{n}}{\Delta x}&=\tilde{w}_x(n\Delta t, i\Delta x, k\Delta r)+\frac12\tilde{w}_{xx}(n\Delta t, i\Delta x, k\Delta r)\Delta x\\
  &\quad+\>o(\Delta x),\\
   \frac{\tilde{w}_{i,k}^{n}-\tilde{w}_{i-1,k}^{n}}{\Delta x}&=\tilde{w}_x(n\Delta t, (i-1)\Delta x, k\Delta r)\\
   &\quad+\>\frac12\tilde{w}_{xx}(n\Delta t, (i-1)\Delta x, k\Delta r)\Delta x+o(\Delta x),\\
  \frac{\tilde{w}_{i,k}^{n+1}-\tilde{w}_{i,k}^n }{\Delta_t}&=\tilde{w}_t((n+1)\Delta t, i\Delta x, k\Delta r)+\frac12\tilde{w}_{tt}(n\Delta t, i\Delta x, k\Delta_r)\Delta t\\
  &\quad+\>o(\Delta t).
 \end{align*}
 Hence, the consistency of the scheme follows from the continuity of the auxiliary HJB operator (note that the truncation error is at most of order one in each parameter, for the approximation of the first derivatives).\\
\indent We now prove the local stability. To this end, set $\I_O:=\{-p,\dots,p\}\times\{-q,\dots,q\}$, where $p,q\in\N$ are the largest possible natural numbers such that $$[-p\Delta x, p\Delta x]\times [-q\Delta r,q\Delta r]\subset P_r(O),$$ with $P_r$ denoting the orthogonal projection of $]0,T]\times\R\times\R$ on $\R\times\R$. Using Assumption \ref{lo}, we can write
\begin{IEEEeqnarray*}{rCl}
\big|\tilde{w}_{i,k}^{n+1}\big|&=&\bigg|\tilde{w}_{i,k}^n-\frac{\Delta t}2\bigg(\frac{i\Delta x\sigma}{\Delta r}\bigg)^2 \big(\tilde{w}_{i,k}^n -\tilde{w}_{i,k-1}^n +\tilde{w}_{i,k}^n-\tilde{w}_{i,k+1}^n -\>\big(\tilde{w}_{i,k}^n -\tilde{w}_{i,k+1}^n \big)^2\big)\\
&&+\>\frac{\Delta r\Delta t}{(\Delta x)^2}\frac{\max\big(\tilde{w}_{i,k}^n -\tilde{w}_{i-1,k}^n,\tilde{w}_{i,k}^n -\tilde{w}_{i+1,k}^n,0 \big)^2 }{\tilde{w}_{i,k}^n -\tilde{w}_{i,k-1}^n }\bigg|\\                 
&\leq& \big|\tilde{w}_{i,k}^n\big|+ \frac{\Delta t}2\big(i\Delta x\sigma\big)^2 \bigg|\frac{\tilde{w}_{i,k}^n -\tilde{w}_{i,k-1}^n +\tilde{w}_{i,k}^n-\tilde{w}_{i,k+1}^n}{\Delta r} -\bigg(\frac{\tilde{w}_{i,k}^n -\tilde{w}_{i,k+1}^n}{\Delta r} \bigg)^2\bigg|\\
&&-\>\frac{\Delta r\Delta t}{(\Delta x)^2}\frac{\max\big(\tilde{w}_{i,k}^n -\tilde{w}_{i-1,k}^n,\tilde{w}_{i,k}^n -\tilde{w}_{i+1,k}^n,0 \big)^2 }{\tilde{w}_{i,k}^n -\tilde{w}_{i,k-1}^n }\\
&\leq& \big|\tilde{w}_{i,k}^n\big|+ \Delta t\big(X_O\sigma\big)^2(\overline{K}'_O+\overline{K}^2_O)\\
&&+\>\frac{\Delta t}{(\Delta x)^2}\frac{\max\big(\tilde{w}_{i,k}^n -\tilde{w}_{i-1,k}^n,\tilde{w}_{i,k}^n -\tilde{w}_{i+1,k}^n,0 \big)^2 }{K_O }\\
&\leq& \big|\tilde{w}_{i,k}^n\big|+ \Delta t\big(X_O\sigma\big)^2(\overline{K}'_O+\overline{K}^2_O)+\frac{\Delta t L_O^2}{K_O},
\end{IEEEeqnarray*}
which implies that
\begin{IEEEeqnarray*}{rCl}
\max_{i,k\in\I_O}\big|\tilde{w}_{i,k}^{n+1}\big|&\leq& \max_{i,k\in\I_O}\big|\tilde{w}_{i,k}^0\big|+ n\Delta t\big(X_O\sigma\big)^2(\overline{K}'_O+\overline{K}^2_O) +n\frac{\Delta t L_O^2}{K_O}\\
&=&\max_{i,k\in\I_O}\big|\tilde{w}_{i,k}^0\big|+ T\big(X_O\sigma\big)^2(\overline{K}'_O+\overline{K}^2_O) +\frac{T L_O^2}{K_O}\\
&<& \infty,
\end{IEEEeqnarray*}
and this proves the stability of the scheme. We have thus established that this explicit scheme converges to the viscosity solution of \eqref{hab}.

Let us now consider the more general case. We will need the following lemma.
\begin{Lem}\label{sgi}
Take $X\in\R^d$. Then  the map
\begin{equation}
\begin{array}{rll}
]0,\infty[&\longrightarrow& \R\\
\widetilde{f}^*_X:T&\longmapsto &Tf^*\big(-\frac{X}T\big)\\
\end{array}
\end{equation}
is strictly decreasing in $T$.
\end{Lem}
\begin{proof}
First note that, due to the strict convexity of $f$, $f^*$ is also strictly convex and hence  fulfills the following subgradient inequality,
\[
f^*(b)-f^*(a)>(b- a)\cdot \nabla f^*(a).
\]
Setting now $b=0$ in the preceding inequality,  we get 
\begin{equation}
 a \cdot\nabla f^*(a)> f^*(a)\geq 0,\label{ssgi}
\end{equation}
because $f^*(0)=0$.
Computing the derivative of $\widetilde{f}^*_X$ with respect to $T$ we obtain
\[
\widetilde{f^*}'_X(T)=f^*\Big(-\frac{X}T\Big)-\frac{X}T\nabla_x f^*\Big(-\frac{X}T\Big),
\]
which is strictly negative, due to the preceding subgradient inequality.
\end{proof}
Suppose that $f$ is symmetric (i.e., $f(x)=f(-x),\>\; \forall x\in\R$) and $b\neq0$. (Note that this symmetry also holds for $f^*$). Since $f^*(w_x)=f^*(|w_x|)$, we obtain the following expression (scheme) for the term $f^*(w_x/w_r)$:
\[
f^*\bigg(\frac{\Delta r}{\Delta x}\cdot\frac{\max\big(\tilde{w}_{i,k}^{n} -\tilde{w}_{i-1,k}^{n},\tilde{w}_{i,k}^{n} -\tilde{w}_{i+1,k}^{n} \big) }{ \tilde{w}_{i,k}^{n} -\tilde{w}_{i,k-1}^{n} }\bigg).
\]
Therefore we can derive the following generalization of the preceding scheme:
\begin{IEEEeqnarray*}{rCl}
\IEEEeqnarraymulticol{3}{l}{S(h,\Delta t,\Delta x,\Delta r,\tilde{w}_{i,k}^{n+1},[\tilde{w}_{i+1,k}^{n}, \tilde{w}^{n}_{i-1,k}, \tilde{w}_{i,k+1}^{n}, \tilde{w}^{n}_{i,k-1}, \tilde{w}_{i,k}^{n}])}\\
&=&\frac{\tilde{w}_{i,k}^{n+1}-\tilde{w}_{i,k}^n}{\Delta t}+\frac12\bigg(\frac{i\Delta x\sigma}{\Delta r}\bigg)^2 \big(\tilde{w}_{i,k}^{n} -\tilde{w}_{i,k-1}^{n} +\tilde{w}_{i,k}^{n}-\tilde{w}_{i,k+1}^{n} \\
&&-\>\big(\tilde{w}_{i,k}^{n} -\tilde{w}_{i,k+1}^{n} \big)^2\big)-b\cdot i\,\Delta x\, F_{b,k}(\tilde{w}_{i,k}^n)\\
&&-\>\frac{\tilde{w}_{i,k}^{n} -\tilde{w}_{i,k-1}^{n} }{\Delta r}f^*\bigg(\frac{\Delta r}{\Delta x}\frac{\max\big(\tilde{w}_{i,k}^{n} -\tilde{w}_{i-1,k}^{n},\tilde{w}_{i,k}^{n} -\tilde{w}_{i+1,k}^{n},0 \big) }{ \tilde{w}_{i,k}^{n} -\tilde{w}_{i,k-1}^{n} }\bigg),\\
 \tilde{w}^0_{i,k}&=&0,                      
\end{IEEEeqnarray*}
where 
\begin{equation*}
F_{b,k}(\tilde{w}_{i,k}^n)=\begin{cases}
                                             \frac{\tilde{w}_{i,k+1}^{n} -\tilde{w}_{i,k}^{n} }{\Delta r},& \text{if}\; \sgn(b\cdot i)>0,\\
                                            \frac{\tilde{w}_{i,k}^{n} -\tilde{w}_{i,k-1}^{n} }{\Delta r},& \text{if}\; \sgn(b\cdot i)\leq0.\\
                      \end{cases}  
\end{equation*}         
Using Lemma \ref{sgi}, we have 
\[
-\>\frac{\tilde{w}_{i,k}^{n} -\tilde{w}_{i,k-1}^{n} }{\Delta r}f^*\bigg(\frac{\Delta r}{\Delta x}\frac{\max\big(\tilde{w}_{i,k}^{n} -\tilde{w}_{i-1,k}^{n},\tilde{w}_{i,k}^{n} -\tilde{w}_{i+1,k}^{n},0 \big) }{ \tilde{w}_{i,k}^{n} -\tilde{w}_{i,k-1}^{n} }\bigg)
\]
is nonincreasing in $\tilde{w}_{i,k-1}^{n}$. Due to the definition of $F_{b,k},$ it is also nonincreasing in $\tilde{w}_{i,k-1}^{n}$,  and the scheme is hence unconditionally nonincreasing in this parameter. 
Noting that $\tilde{w}_{i,k}^{n} -\tilde{w}_{i,k-1}^{n}<0$, and using the fact that $f^*$ is decreasing on $]-\infty,0]$ and increasing on $[0,\infty[$ (due to its positivity, convexity and the fact that $f^*(0)=0$), it follows that the scheme is nonincreasing in both $\tilde{w}_{i-1,k}^{n}$ and $ \tilde{w}_{i+1,k}^{n}$. Again, the definition of $F_{b,k}$ and the same argumentation as before (for $|\tilde{w}^{n}_{i,k}-\tilde{w}^{n}_{i,k+1}|\leq1/2$, as seen above) allow us to deduce that the scheme is nonincreasing in $ \tilde{w}_{i,k+1}^{n}$. We now present a sufficient condition under which $S$ is nonincreasing in $\tilde{w}^{n}_{i,k}$. \\
\underline{First case:}  $\max\big(\tilde{w}_{i,k}^n -\tilde{w}_{i-1,k}^n,\tilde{w}_{i,k}^n -\tilde{w}_{i+1,k}^n,0 \big)=0$.\\
Consider the function
$$ \psi_3: \tilde{w}^{n}_{i,k}\longmapsto -\tilde{w}^{n}_{i,k}+ \frac{\Delta t}2\bigg(\frac{i\Delta x\sigma}{\Delta r}\bigg)^2 \big(\tilde{w}^{n}_{i,k}-(\tilde{w}^{n}_{i,k}-\tilde{w}^{n}_{i,k+1})^2\big)-b\cdot i\,\Delta x\Delta t\, F_{b,k}(\tilde{w}_{i,k}^n).$$
Its derivative is given by 
$$ \psi'_3: \tilde{w}^{n}_{i,k}\longmapsto -1+ \frac{\Delta t}2\bigg(\frac{i\Delta x\sigma}{\Delta r}\bigg)^2 \big(1-2(\tilde{w}^{n}_{i,k}-\tilde{w}^{n}_{i,k+1})\big)+|b\cdot i\Delta x|\,\frac{\Delta t}{\Delta r}.$$
Then, for
$$|\tilde{w}^{n}_{i,k}-\tilde{w}^{n}_{i,k+1}|\leq1/2\quad\text{and}\quad\Delta t\Big(\frac32\Big(\frac{i\Delta x\sigma}{\Delta r}\Big)^2+\frac{|b\cdot i\Delta x|}{\Delta r}\Big)\leq1,$$ we have that $\psi'_3\leq0$, and  $S$ is hence nonincreasing in $\tilde{w}_{i,k}^n$.\\
\underline{Second case:}  $\max\big(\tilde{w}_{i,k}^n -\tilde{w}_{i-1,k}^n,\tilde{w}_{i,k}^n -\tilde{w}_{i+1,k}^n,0 \big)\neq 0$.\\
We can suppose w.l.o.g. that $\max\big(\tilde{w}_{i,k}^n -\tilde{w}_{i-1,k}^n,\tilde{w}_{i,k}^n -\tilde{w}_{i+1,k}^n,0 \big)=\tilde{w}_{i,k}^n -\tilde{w}_{i-1,k}^n$. Consider now the following function
\begin{IEEEeqnarray*}{rCl} 
\IEEEeqnarraymulticol{3}{l}{\psi_4: \tilde{w}^{n}_{i,k}\longmapsto -\tilde{w}^{n}_{i,k}+ \frac{\Delta t}2\bigg(\frac{i\Delta x\sigma}{\Delta r}\bigg)^2 \big(\tilde{w}^{n}_{i,k}-(\tilde{w}^{n}_{i,k}-\tilde{w}^{n}_{i,k+1})^2\big)
-b\cdot i\,\Delta x\Delta t\, F_{b,k}(\tilde{w}_{i,k}^n)}\\
&&-\>\Delta t\frac{\tilde{w}_{i,k}^{n} -\tilde{w}_{i,k-1}^{n} }{\Delta r}f^*\bigg(\frac{\Delta r}{\Delta x}\frac{\tilde{w}_{i,k}^n -\tilde{w}_{i-1,k}^n}{ \tilde{w}_{i,k}^{n} -\tilde{w}_{i,k-1}^{n} }\bigg),
\end{IEEEeqnarray*}
whose derivative is given by
\begin{IEEEeqnarray*}{rCl} 
\IEEEeqnarraymulticol{3}{l}{\psi'_4: \tilde{w}^{n}_{i,k}\longmapsto -1+ \frac{\Delta t}2\bigg(\frac{i\Delta x\sigma}{\Delta r}\bigg)^2 \big(1-2(\tilde{w}^{n}_{i,k}-\tilde{w}^{n}_{i,k+1})\big)+|b\cdot i\Delta x|\,\frac{\Delta t}{\Delta r}}\\
&&-\>\frac{\Delta t}{\Delta r}\bigg( f^*\bigg(\frac{\Delta r}{\Delta x}\frac{\tilde{w}_{i,k}^n -\tilde{w}_{i-1,k}^n }{ \tilde{w}_{i,k}^{n} -\tilde{w}_{i,k-1}^{n} }\bigg)\\
&&+\>\frac{\Delta r}{\Delta x}\frac{ \tilde{w}_{i,k}^{n} -\tilde{w}_{i,k-1}^{n}-\tilde{w}_{i,k}^n +\tilde{w}_{i-1,k}^n}{\tilde{w}_{i,k}^{n} -\tilde{w}_{i,k-1}^{n}} (f^*)'\bigg(\frac{\Delta r}{\Delta x}\frac{\tilde{w}_{i,k}^n -\tilde{w}_{i-1,k}^n }{ \tilde{w}_{i,k}^{n} -\tilde{w}_{i,k-1}^{n} }\bigg)\bigg).
\end{IEEEeqnarray*}
As in the preceding special case, taking  $h>0$ such that  $|\tilde{w}^{m}_{j,l}-\tilde{w}^{m'}_{j',l'}|\leq1/2$, for $|(m,j,l)-(m',j',l')|\leq h$, and using the fact that $(f^*)'$ is negative on $]-\infty,0[$, nonnegative otherwise and  decreasing on the whole of $\R$, we can
write 
\begin{IEEEeqnarray*}{rCl} 
\IEEEeqnarraymulticol{3}{l}{f^*\bigg(\frac{\Delta r}{\Delta x}\frac{\tilde{w}_{i,k}^n -\tilde{w}_{i-1,k}^n }{ \tilde{w}_{i,k}^{n} -\tilde{w}_{i,k-1}^{n} }\bigg)}\\
&&+\>\frac{\Delta r}{\Delta x}\frac{ \tilde{w}_{i,k}^{n} -\tilde{w}_{i,k-1}^{n}-\tilde{w}_{i,k}^n +\tilde{w}_{i-1,k}^n}{\tilde{w}_{i,k}^{n} -\tilde{w}_{i,k-1}^{n}} (f^*)'\bigg(\frac{\Delta r}{\Delta x}\frac{\tilde{w}_{i,k}^n -\tilde{w}_{i-1,k}^n }{ \tilde{w}_{i,k}^{n} -\tilde{w}_{i,k-1}^{n} }\bigg)\\
&\geq&\>\frac{\Delta r}{\Delta x}\frac{ \tilde{w}_{i,k}^{n} -\tilde{w}_{i,k-1}^{n}-\tilde{w}_{i,k}^n +\tilde{w}_{i-1,k}^n}{\tilde{w}_{i,k}^{n} -\tilde{w}_{i,k-1}^{n}} (f^*)'\bigg(\frac{\Delta r}{\Delta x}\frac{\tilde{w}_{i,k}^n -\tilde{w}_{i-1,k}^n }{ \tilde{w}_{i,k}^{n} -\tilde{w}_{i,k-1}^{n} }\bigg)\\
&\geq&\>-\frac1{\Delta x K_O} (f^*)'\Big(\frac1{2 \Delta x K_O }\Big).
\end{IEEEeqnarray*}
Finally, we get with the CFL condition
$$\Delta t\Big(\frac32\Big(\frac{i\Delta x\sigma}{\Delta r}\Big)^2+\frac{|b\cdot i\Delta x|}{\Delta r}+\frac1{\Delta x \Delta r K_O} (f^*)'\Big(\frac1{2 \Delta x K_O }\Big)\Big)\leq1$$
that $\psi_4'(w_{i,k}^n)\leq 0$, and the scheme is therefore locally $\delta$-monotone.\\
The consistency of the scheme can be proved in an analogous manner as above,  using the preceding Taylor expansions and the fact that both  $\max$ and $f^*$ are continuous functions.\\
We have now left to prove  the local stability. But here again, using Assumption \ref{lo} we get
\begin{IEEEeqnarray*}{rCl}
\big|\tilde{w}_{i,k}^{n+1}\big|&=&\bigg|\tilde{w}_{i,k}^n-\frac{\Delta t}2\bigg(\frac{i\Delta x\sigma}{\Delta r}\bigg)^2 \big(\tilde{w}_{i,k}^n -\tilde{w}_{i,k-1}^n +\tilde{w}_{i,k}^n-\tilde{w}_{i,k+1}^n -\>\big(\tilde{w}_{i,k}^n -\tilde{w}_{i,k+1}^n \big)^2\big)\\
&&-\>b\cdot i\,\Delta x \Delta t\, F_{b,k}(\tilde{w}_{i,k}^n)\\
&&-\>\Delta t\frac{\tilde{w}_{i,k}^{n} -\tilde{w}_{i,k-1}^{n} }{\Delta r}f^*\bigg(\frac{\Delta r}{\Delta x}\frac{\max\big(\tilde{w}_{i,k}^{n} -\tilde{w}_{i-1,k}^{n},\tilde{w}_{i,k}^{n} -\tilde{w}_{i+1,k}^{n},0 \big) }{ \tilde{w}_{i,k}^{n} -\tilde{w}_{i,k-1}^{n} }\bigg)\\                 
&\leq& \big|\tilde{w}_{i,k}^n\big|+ \frac{\Delta t}2\bigg(\frac{i\Delta x\sigma}{\Delta r}\bigg)^2 \big|\tilde{w}_{i,k}^n -\tilde{w}_{i,k-1}^n +\tilde{w}_{i,k}^n-\tilde{w}_{i,k+1}^n -\>\big(\tilde{w}_{i,k}^n -\tilde{w}_{i,k+1}^n \big)^2\big|\\
&&-\>|b\cdot i\,\Delta x|\Delta t\, F_{b,k}(\tilde{w}_{i,k}^n)\\
&&-\>\Delta t\frac{\tilde{w}_{i,k}^{n} -\tilde{w}_{i,k-1}^{n} }{\Delta r}f^*\bigg(\frac{\Delta r}{\Delta x}\frac{\max\big(\tilde{w}_{i,k}^{n} -\tilde{w}_{i-1,k}^{n},\tilde{w}_{i,k}^{n} -\tilde{w}_{i+1,k}^{n},0 \big) }{ \tilde{w}_{i,k}^{n} -\tilde{w}_{i,k-1}^{n} }\bigg)\\   
&\leq& \big|\tilde{w}_{i,k}^n\big|+ \Delta t\big(X_O\sigma\big)^2(\overline{K}'_O+\overline{K}^2_O)+ |b|\big|X_O\big|\Delta t\overline{K}_O+\overline{K}_O\Delta t f^*\bigg(\frac{L_O}{K_O}\bigg),
 \end{IEEEeqnarray*}
which gives us recursively 
\begin{align*}
\max_{i,k\in\I_O} |\tilde{w}_{i,k}^{n+1}|&\leq\max_{i,k\in\I_0} |\tilde{w}_{i,k}^0|+ T\bigg(\big(X_O\sigma\big)^2(\overline{K}'_O+\overline{K}^2_O)
+ |b|\big|X_O\big|\overline{K}_O+\overline{K}_O f^*\bigg(\frac{L_O}{K_O}\bigg)\bigg)\\
&<\infty.
\end{align*}
Thus, the local stability is proved. This establishes that the preceding explicit scheme indeed converges to the viscosity solution.

\subsection{Construction of a converging implicit scheme}
Proving the $\delta$-monotonicity will turn out to be more obvious  for the following implicit scheme  than for the preceding explicit one. Moreover, the following implicit scheme will be  unconditionally stable. Nevertheless, there will be two main issues which restrict its use. The first one follows from the fact that terms must be obtained by implicit computations, which implies that we have to find them before using them in the scheme (by applying in general a Newton-Raphson method). In this nonlinear case, this will result in an  implementation error, which will  be combined with the  approximation error. The second issue follows from the fact that the local stability is difficult to obtain in practice (due to the appearance of a quotient term and the difficulty of computing the constants $\overline{K_O}$ and $L_{O}$, which will moreover impose restrictions on $\Delta x$ and $\Delta r$), as we will see below.
Let us consider the following scheme, where $b=0$, $f(x)=\lambda x^2,$ and $\lambda>0$.
\begin{IEEEeqnarray*}{rCl}
\IEEEeqnarraymulticol{3}{l}{S(h,\Delta t,\Delta x,\Delta r,\tilde{w}_{i,k}^{n+1},[\tilde{w}_{i+1,k}^{n+1}, \tilde{w}^{n+1}_{i-1,k}, \tilde{w}_{i,k+1}^{n+1}, \tilde{w}^{n+1}_{i,k-1}, \tilde{w}_{i,k}^{n}])}\\
&=&\frac{\tilde{w}_{i,k}^{n+1}-\tilde{w}_{i,k}^n}{\Delta t}+\frac12\bigg(\frac{i\Delta x\sigma}{\Delta r}\bigg)^2 \big(\tilde{w}_{i,k}^{n+1} -\tilde{w}_{i,k-1}^{n+1} +\tilde{w}_{i,k}^{n+1}-\tilde{w}_{i,k+1}^{n+1} \\
&&-\>\big(\tilde{w}_{i,k}^{n+1} -\tilde{w}_{i,k+1}^{n+1} \big)^2\big)-\>\frac{\Delta r}{4\lambda(\Delta x)^2}\cdot\frac{\max\big(\tilde{w}_{i,k}^{n+1} -\tilde{w}_{i-1,k}^{n+1},\tilde{w}_{i,k}^{n+1} -\tilde{w}_{i+1,k}^{n+1},0 \big)^2 }{\tilde{w}_{i,k}^{n+1} -\tilde{w}_{i,k-1}^{n+1} },\\
\tilde{w}^0_{i,k}&=&0,                                        
\end{IEEEeqnarray*}
whose stencil is represented below as:\\
\begin{tikzpicture}
\matrix (m) [matrix of math nodes, row sep=3em,
column sep=3em]{
\tilde{w}^{n}_{i,k}& \tilde{w}^{n+1}_{i,k+1}& &  \\
\tilde{w}^{n+1}_{i-1,k} &\tilde{w}^{n+1}_{i,k} & \tilde{w}^{n+1}_{i+1,k} & \\
& \tilde{w}^{n+1}_{i,k-1} & \tilde{w}^{n+1}_{i,k}&  \\};
 \path
 (m-1-1) [densely dotted] edge (m-2-2);
 \path
(m-1-2) edge (m-2-2)
(m-2-1) edge (m-2-2)
(m-2-2)edge (m-2-3)
(m-3-2)edge (m-2-2);
 \path[-stealth]
 (m-2-2) edge (m-3-3);
\end{tikzpicture}

First, note that 
\[
-\frac{\Delta r}{4\lambda(\Delta x)^2}\cdot\frac{\max\big(\tilde{w}_{i,k}^{n+1} -\tilde{w}_{i-1,k}^{n+1},\tilde{w}_{i,k}^{n+1} -\tilde{w}_{i+1,k}^{n+1},0 \big)^2 }{ \tilde{w}_{i,k}^{n+1} -\tilde{w}_{i,k-1}^{n+1} }
\]
is nonincreasing in both $\tilde{w}_{i,k-1}^{n+1},\tilde{w}_{i-1,k}^{n+1}$ and $ \tilde{w}_{i+1,k}^{n+1}$. Take now $h>0$ small enough such that $|\tilde{w}^{n+1}_{i,k}-\tilde{w}^{n+1}_{i,k+1}|\leq 1/2$. Then,
\[
\tilde{w}_{i,k}^{n+1}-\tilde{w}_{i,k+1}^{n+1}-(\tilde{w}_{i,k}^{n+1} -\tilde{w}_{i,k+1}^{n+1} \big)^2
\]
is nonincreasing in $\tilde{w}_{i,k}^{n+1} -\tilde{w}_{i,k-1}^{n+1}$, because the function $x-x^2$ is increasing in $x$ for $-1/2\leq x\leq 1/2$. Since the first term is also nonincreasing in  $\tilde{w}_{i,k}^n$, we have thus proved that the scheme is (unconditionally) monotone, for  $|w^{n+1}_{i,k}-w^{n+1}_{i,k+1}|\leq 1/2$. It remains to prove its consistency and local stability.  Classical computations using the Taylor expansion again yield: 
\begin{align*}
 \frac{\tilde{w}_{i,k+1}^{n+1}+\tilde{w}^{n+1}_{i,k-1}-2\tilde{w}_{i,k}^{n+1}}{(\Delta r)^2}&=\tilde{w}_{rr}((n+1)\Delta t, i\Delta x, k\Delta r)\\
 &\quad+\>\frac1{12}\tilde{w}_{rrrr}((n+1)\Delta t, i\Delta x, k\Delta r) (\Delta r)^2+ o(\Delta r)^2,\\
 \frac{\tilde{w}_{i,k+1}^{n+1}-\tilde{w}_{i,k}^{n+1} }{\Delta r}&=\tilde{w}_r((n+1)\Delta t, i\Delta x, k\Delta r)+\frac12\tilde{w}_{rr}((n+1)\Delta t, i\Delta x, k\Delta_r)\Delta r\\
 &\quad+\>o(\Delta r),\\
 \frac{\tilde{w}_{i,k}^{n+1}-\tilde{w}_{i,k-1}^{n+1} }{\Delta r}&=\tilde{w}_r((n+1)\Delta t, i\Delta x, (k-1)\Delta r)\\
          &\quad+\>\frac12\tilde{w}_{rr}((n+1)\Delta t, i\Delta x, (k-1)\Delta_r)\Delta r+o(\Delta r),\\
  \frac{\tilde{w}_{i+1,k}^{n+1}-\tilde{w}_{i,k}^{n+1}}{\Delta x}&=\tilde{w}_x((n+1)\Delta t, i\Delta x, k\Delta r)+\frac12\tilde{w}_{xx}((n+1)\Delta t, i\Delta x, k\Delta r)\Delta x\\
  &\quad+\>o(\Delta x),\\
   \frac{\tilde{w}_{i,k}^{n+1}-\tilde{w}_{i-1,k}^{n+1}}{\Delta x}&=\tilde{w}_x((n+1)\Delta t, (i-1)\Delta x, k\Delta r)\\
   &\quad+\>\frac12\tilde{w}_{xx}((n+1)\Delta t, (i-1)\Delta x, k\Delta r)\Delta x+o(\Delta x),\\
  \frac{\tilde{w}_{i,k}^{n+1}-\tilde{w}_{i,k}^n }{\Delta_t}&=\tilde{w}_t((n+1)\Delta t, i\Delta x, k\Delta r)+\frac12\tilde{w}_{tt}((n+1)\Delta t, i\Delta x, k\Delta_r)\Delta t\\
  &\quad+\>o(\Delta t).
 \end{align*}
 Note that here again the truncation error is only of order one in each parameter for the approximation of the first derivatives. However, this order will be weakened because of  implicit computation of the corresponding terms. Hence, the consistency of the scheme follows from the continuity of the auxiliary HJB operator. To prove its local stability, we have to require that $\sigma\Delta x/\Delta r$ is bounded. We use the fact that 
\[
\max\bigg\{\bigg| \frac{\tilde{w}_{i,k+1}^{n+1}-\tilde{w}_{i,k}^{n+1} }{\Delta r}\bigg|,\bigg|\frac{\tilde{w}_{i,k}^{n+1}-\tilde{w}_{i,k-1}^{n+1} }{\Delta r}\bigg|\bigg\}\leq\overline{K}_O,
\]
on a bounded open set $O$. Due to Assumption \ref{lo}, we also have that
\[
\max\bigg\{\bigg| \frac{\tilde{w}_{i,k}^{n+1}-\tilde{w}_{i-1,k}^{n+1} }{\Delta x}\bigg|,\bigg|\frac{\tilde{w}_{i,k}^{n+1}-\tilde{w}_{i+1,k}^{n+1} }{\Delta x}\bigg|\bigg\}\leq L_O.
\]
 Hence,  expressing the differences as follows:
\begin{IEEEeqnarray*}{rCl}
\tilde{w}_{i,k}^{n+1}-\tilde{w}_{i,k}^n&=&\frac{\Delta t}2\bigg(\frac{i\Delta x\sigma}{\Delta r}\bigg)^2 \big(\tilde{w}_{i,k}^{n+1} -\tilde{w}_{i,k-1}^{n+1} +\tilde{w}_{i,k}^{n+1}-\tilde{w}_{i,k+1}^{n+1}\\
&& -\>\big(\tilde{w}_{i,k}^{n+1} -\tilde{w}_{i,k+1}^{n+1},0 \big)^2\big)\\
&&-\>\frac{\Delta t\Delta r}{(\Delta x)^2}\frac{\max\big(\tilde{w}_{i,k}^{n+1} -\tilde{w}_{i-1,k}^{n+1},\tilde{w}_{i,k}^{n+1} -\tilde{w}_{i+1,k}^{n+1},0 \big)^2 }{\tilde{w}_{i,k}^{n+1} -\tilde{w}_{i,k-1}^{n+1} },\\
\end{IEEEeqnarray*}
 we finally deduce
\begin{align*}
\max_{i,k\in\I_O} |\tilde{w}_{i,k}^{n+1}|&\leq\max_{i,k\in\I_0} |\tilde{w}_{i,k}^0|+\frac{3n\Delta t}8\bigg(\frac{i\Delta x\sigma}{\Delta r}\bigg)^2+n\Delta t\frac{ L_O^2}{\overline{K}_O }\\
&=\max_{i,k\in\I_0} |\tilde{w}_{i,k}^0|+\frac{3T}8\bigg(\frac{i\Delta x\sigma}{\Delta r}\bigg)^2+T\frac{ L_O^2}{\overline{K}_O}\\
&\leq \max_{i,k\in\I_0} |\tilde{w}_{i,k}^0|+\frac{3T|I_O|}8\bigg(\frac{\Delta x\sigma}{\Delta r}\bigg)^2+T\frac{ L_O^2}{\overline{K}_O} \\
&\leq \max_{i,k\in\I_0} |\tilde{w}_{i,k}^0|+\frac{3T|I_O|}8C^2+T\frac{ L_O^2}{\overline{K}_O},
\end{align*}
where $C\geq\sigma\Delta x/\Delta r$. Hence, this proves the stability of the scheme. Thus, the implicit scheme considered converges to the viscosity solution of \eqref{hab}.

In a more general framework (i.e. $b\neq0$ and $f$ symmetric), as it was the case with the explicit scheme above, we can consider the following scheme:
\begin{IEEEeqnarray*}{rCl}
\IEEEeqnarraymulticol{3}{l}{S(h,\Delta t,\Delta x,\Delta r,\tilde{w}_{i,k}^{n+1},[\tilde{w}_{i+1,k}^{n+1}, \tilde{w}^{n+1}_{i-1,k}, \tilde{w}_{i,k+1}^{n+1}, \tilde{w}^{n+1}_{i,k-1}, \tilde{w}_{i,k}^{n}])}\\
&=&\frac{\tilde{w}_{i,k}^{n+1}-\tilde{w}_{i,k}^n}{\Delta t}+\frac12\bigg(\frac{i\Delta x\sigma}{\Delta r}\bigg)^2 \big(\tilde{w}_{i,k}^{n+1} -\tilde{w}_{i,k-1}^{n+1} +\tilde{w}_{i,k}^{n+1}-\tilde{w}_{i,k+1}^{n+1} \\
&&-\>\big(\tilde{w}_{i,k}^{n+1} -\tilde{w}_{i,k+1}^{n+1} \big)^2\big)-b\cdot i\,\Delta x\, F_{b,k}(\tilde{w}_{i,k}^{n+1})\\
&&-\>\frac{\tilde{w}_{i,k}^{n+1} -\tilde{w}_{i,k-1}^{n+1} }{\Delta r}f^*\bigg(\frac{\Delta r}{\Delta x}\frac{\max\big(\tilde{w}_{i,k}^{n+1} -\tilde{w}_{i-1,k}^{n+1},\tilde{w}_{i,k}^{n+1} -\tilde{w}_{i+1,k}^{n+1},0 \big) }{ \tilde{w}_{i,k}^{n+1} -\tilde{w}_{i,k-1}^{n+1} }\bigg),\\
 \tilde{w}^0_{i,k}&=&0,                      
\end{IEEEeqnarray*}
where 
\begin{equation*}
F_{b,k}(\tilde{w}_{i,k}^{n+1})=\begin{cases}
                                             \frac{\tilde{w}_{i,k+1}^{n+1} -\tilde{w}_{i,k}^{n+1} }{\Delta r},& \text{if}\; \sgn(b\cdot i)>0,\\
                                            \frac{\tilde{w}_{i,k}^{n+1} -\tilde{w}_{i,k-1}^{n+1} }{\Delta r},& \text{if}\; \sgn(b\cdot i)\leq0.\\
                      \end{cases}  
\end{equation*}         
In analogy to the previous argumentation, we can prove that this scheme is again (unconditionally) nonincreasing in
 $\tilde{w}_{i-1,k}^{n+1},\;\tilde{w}_{i+1,k}^{n+1},\;\tilde{w}_{i,k-1}^{n+1},\;\tilde{w}_{i,k+1}^{n+1}$ and $\tilde{w}_{i,k}^n$  (when taking $h>0$ small enough such that $|\tilde{w}^{n+1}_{i,k}-\tilde{w}^{n+1}_{i,k+1}|\leq 1/2$), and it is therefore  monotone. \\
The consistency of the scheme can be proved in the same manner as beforehand,  by using the preceding Taylor expansions and the fact that both functions $\max$ and $f^*$ are continuous.\\
Using step by step the arguments and computations used to prove the local stability of the explicit version of this scheme also yields its local stability (where here again we have to impose suitable restrictions on $\Delta x$ and $\Delta r$). Hence, this scheme converges to the unique viscosity solution, provided that a method to compute the implicit terms is given.
\section{Numerical examples}
In this section, we provide an application of the preceding results. Implementing our implicit schemes is a challenging task, due to mainly the following two reasons.  First, classical computations in the spirit of the Newton-Raphson method would become rather involved in our case (because of the nonlinear part). This is due to the fact that, although the quadratic term can be linearized in order to make the task easier, there is still a quotient term to be dealt with. Second, the number of implicit variables to compute at each stage (five terms, as it can be seen in its corresponding stencil above) is another reason why we shall consider here only explicit schemes to visualize the value function of our maximization problem. 

Nevertheless, even in the case of explicit schemes, we still face some issues in our modeling. For example, our initial condition involves exponential growth, which means that taking $T$  small leads to large terms in the exponent. Since in most of the available computer programs we cannot use values larger than $\exp(1000)$, no reasonable results are displayed. For instance,  Matlab displays "\texttt{Inf}"  for $\log(\exp(1000))$, instead of displaying $1000$. Moreover, as we consider only bounded domains for the schemes, we  have to impose boundary conditions, which results in approximation errors. As mentioned above (see Remark \ref{bcs}), we will use the approximated value of $W$ with $R_0$  taken large enough. However,  we cannot take it as large as one wants to (see previous argumentation). Last but not least, the evaluation of the lower bound $K_O$ of the partial derivative $W_r$ presents another issue, since the latter one, which is in general difficult to obtain, is necessary to impose a CFL condition on the grid parameters. \\

\subsection{Exponential value function}

Let us start with approximating a known solution. In particular, we will thus show the accuracy of our scheme. In \citet{SS07}, we have the following explicit formula for the value function of the problem when considering the one-dimensional case with $f(x)=\lambda x^2,\;\lambda>0$, and $u(x)=-\exp(-Ax),\;A>0$:
\begin{equation*}
V(T,X_0,R_0)=-\exp\bigg(-AR_0+X_0^2\sqrt{\frac{\lambda A^3\sigma^2}2}\coth\bigg(T\sqrt{\frac{A\sigma^2}{2\lambda}}\bigg)\bigg).
\end{equation*}
\begin{figure}
\includegraphics[width=16cm, height=9cm]{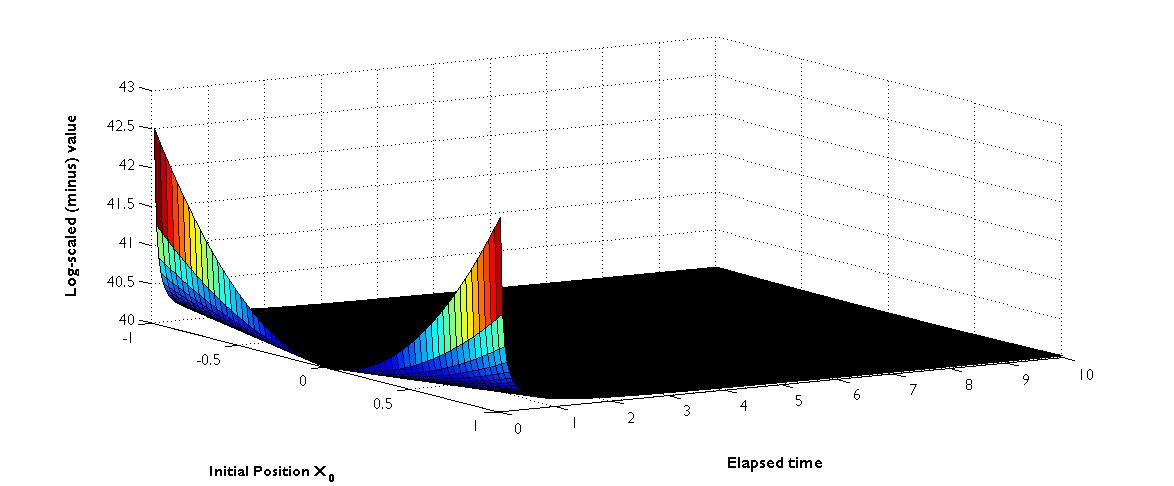}
\caption{Logarithmic representation of the value function (negative values)}
\label{fig 5.1}
\end{figure}
In Figure \ref{fig 5.1}, we show $\log(-V)$ for $R_0=1,\; \lambda=0.1,\;A=1$ and $\sigma=0.1$. We now wish to approximate $$w(T,X_0,R_0):=\log(-V)(T,X_0,R_0)=-AR_0+X_0^2\sqrt{\frac{\lambda A^3\sigma^2}2}\coth\bigg(T\sqrt{\frac{A\sigma^2}{2\lambda}}\bigg),$$    with the help of the following explicit scheme:
\begin{IEEEeqnarray*}{rCl}
\IEEEeqnarraymulticol{3}{l}{S(h,\Delta t,\Delta x,\Delta r,\tilde{w}_{i,k}^{n+1},[\tilde{w}_{i+1,k}^n, \tilde{w}^n_{i-1,k}, \tilde{w}_{i,k+1}^n, \tilde{w}^n_{i,k-1}, \tilde{w}_{i,k}^{n}])}\\
&=&\frac{\tilde{w}_{i,k}^{n+1}-\tilde{w}_{i,k}^n}{\Delta t}+\frac12\bigg(\frac{i\Delta x\sigma}{\Delta r}\bigg)^2 \big(\tilde{w}_{i,k}^n -\tilde{w}_{i,k-1}^n +\tilde{w}_{i,k}^n-\tilde{w}_{i,k+1}^n\\
&& -\>\big(\tilde{w}_{i,k}^n -\tilde{w}_{i,k+1}^n \big)^2-\frac{\Delta r}{4\lambda(\Delta x)^2}\cdot\frac{\max\big(\tilde{w}_{i,k}^n -\tilde{w}_{i-1,k}^n,\tilde{w}_{i,k}^n -\tilde{w}_{i+1,k}^n,0 \big)^2 }{\tilde{w}_{i,k}^n -\tilde{w}_{i,k-1}^n },\\
 \tilde{w}^0_{i,k}-\tilde{u}^0_{i,k}&=&0.                               
\end{IEEEeqnarray*}
\begin{figure}
\centering
\includegraphics[width=6.5in]{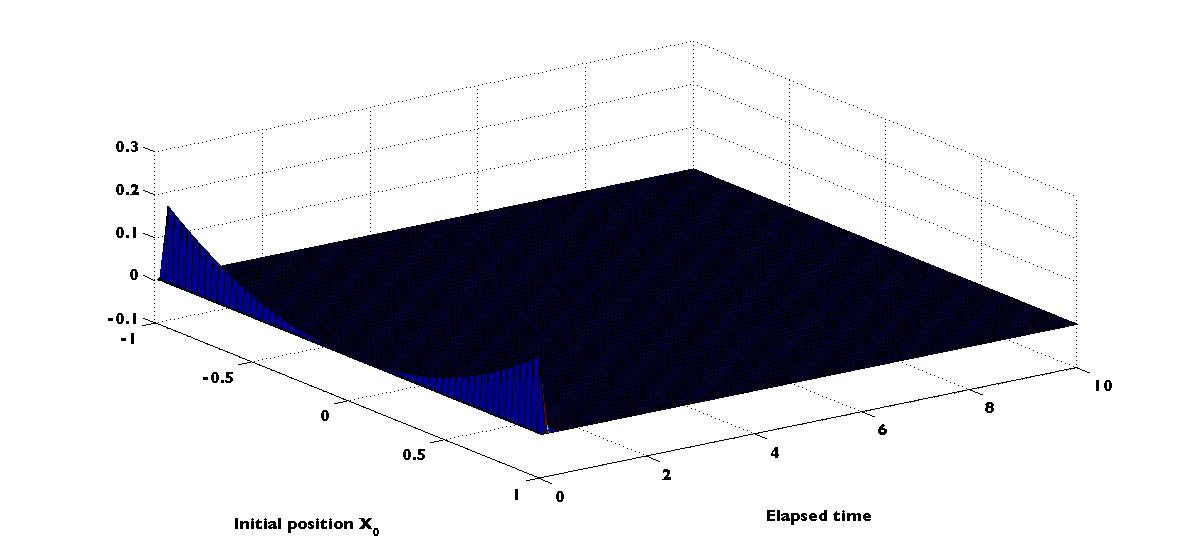}
\caption{Implementation of the real solution in our scheme} 
\label{fig2}
\end{figure}
We cannot directly start with $n=1$ as proposed above, since both $\tilde{w}^0_{i,k}$ and $\tilde{u}^0_{i,k}$ are undefined ($=\infty$), only their differences being defined and equal to 0. Moreover, we will need to impose some boundary conditions.\\
First, note that in this simple case $w_r=-A$, and hence $K_O=A$. Therefore, our CFL condition $\eqref{cfl1}$ is here given by
\begin{equation}
\frac{3\Delta t}{8\lambda}\bigg(\frac{(X_O\sigma)^2(\Delta x)^2A^2+\Delta r}{(\Delta r)^2(\Delta x)^2 A^2}\bigg)\leq 1.\label{cfl2}
\end{equation}
In the following, we will show that the preceding CFL condition was taken rather too restrictive, and our scheme does not need to necessarily fulfill it in order to converge. Subsequently we set
$$O=\;]0.04,10]\;\times\;]-1,1[\;\times\;]-50,50[,\; dr=0.833,\; dt=0.04\;\text{and}\; dx=0.0333.$$ 
We  show the consistency of the scheme by implementing the real solution of \eqref{hjba}, as shown in figure \ref{fig2}. With an absolute value of at most $0.18$, this scheme seems to be very consistent.
Using Proposition \ref{iws}, we set  the following initial condition
\[
w_{i,k}^{1}=\log(B-u(k\Delta r-(i\Delta x)^2/n\Delta t))= -A (k\Delta r+\lambda(i\Delta x)^2/(n\Delta t)).
\]
 We also have to add boundary conditions in our scheme. We define them as follows: denoting by $\pm x_{\max}:=\pm i_{m} \cdot\Delta x$ and $\pm r_{\max}:=\pm k_{m}\cdot \Delta r$ the extreme values taken by $x$ and $r$, respectively, on the grid, we have to set for $n\geq1$:
\[
w_{\pm i_m,k}^{n}=\log(B-u(k\Delta r-(x_{\max})^2/n\Delta t)),
\]
\[
w_{i,\pm k_m}^{n}=\log(B-u(\pm r_{\max}-(i\Delta x)^2/n\Delta t)).
\]
As already argued in Remark \ref{bcs}, this setting could only work out for large values of $R_0$, not for large values of $X_0$. However, in this particular case, this represents a very good setting of the boundary conditions (see figure \ref{fig3}). We also display  the approximation error   (figure \ref{fis}). With at most 2.5 \% error for small time $T$ (and at most 0.03\% from time $T$ larger than 2), our scheme seems to give a very good approximation in this particular case, even if our CFL condition is not satisfied  (the left-hand side term of \eqref{cfl2} yields here 162.0043). \\
\begin{figure}
\raggedleft
\includegraphics[width=1.1\textwidth]{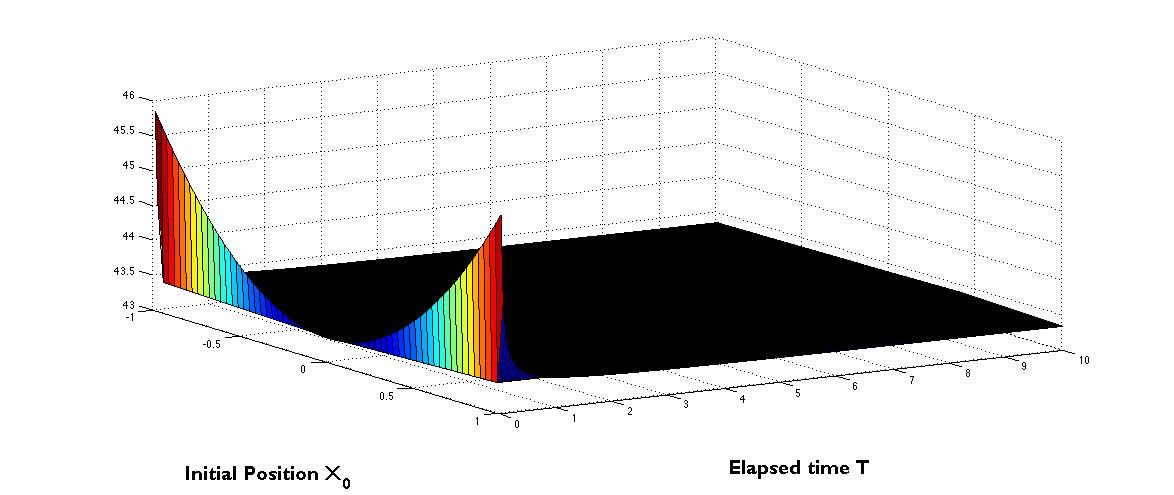}
\caption{Value returned by the scheme} 
\label{fig3}
\end{figure}

\begin{figure}
\raggedleft
\includegraphics[width=1.1\textwidth]{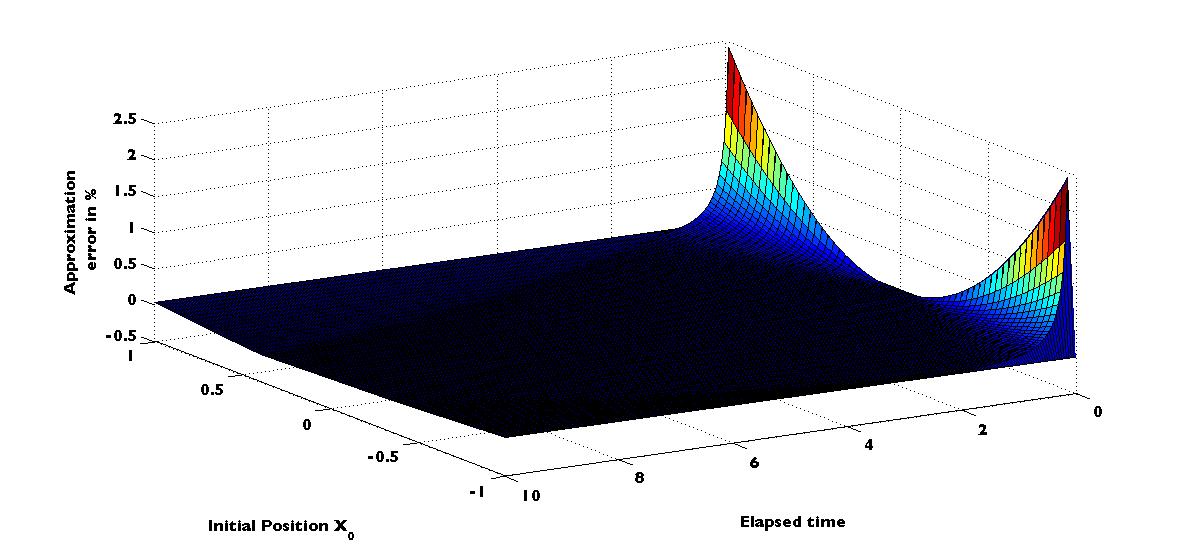}
\caption{Approximation error for $\lambda=0.1,\sigma=0.1,R_0=-43.3333$ and $A=1$} 
\label{fis}
\end{figure}

Nevertheless, things are not working so well when $B$ is not any longer supposed to be equal to zero, since we now have  to deal with the second partial derivative of $w$ in its third parameter  (whereas before it was equal to zero), and since the CFL condition can "explode", due to exponential terms. Let us fix it. To this end, we start with computing a strictly negative upper bound $K_O$, on a bounded set $O$, for $w_r$ (in order to set a CFL condition in our scheme). We compute
\[
w_r( T,X_0,R_0)=\frac{-A\exp\bigg(-AR_0+X_0^2\sqrt{\frac{\lambda A^3\sigma^2}2}\coth\bigg(T\sqrt{\frac{A\sigma^2}{2\lambda}}\bigg)\bigg)}{1+\exp\bigg(-AR_0+X_0^2\sqrt{\frac{\lambda A^3\sigma^2}2}\coth\bigg(T\sqrt{\frac{A\sigma^2}{2\lambda}}\bigg)\bigg)}.
\]
Taking $O=\;]\Delta t;T]\;\times\;]X_{\min};X_{\max}[\;\times\;]R_{\min};R_{\max}[$, and using the fact that $
x\mapsto -\frac{x}{1+x}$ is strictly decreasing for $x>0$,
we infer the following upper bound:
\[
-K_O:=\frac{-A\exp\bigg(-AR_{\max}+x_{\min}^2\sqrt{\frac{\lambda A^3\sigma^2}2}\coth\bigg(T \sqrt{\frac{A\sigma^2}{2\lambda}}\bigg)\bigg)}{1+\exp\bigg(-AR_{\max}+x_{\min}^2\sqrt{\frac{\lambda A^3\sigma^2}2}\coth\bigg(T \sqrt{\frac{A\sigma^2}{2\lambda}}\bigg)\bigg)}\geq w_r( T,X_0,R_0),                           
\]
where $x^2_{\min}:=\inf_{x\in\;]X_{\min};X_{\max}[}x^2$.
Calculating this value of $K_O$  for our parameters $R_{\max}=50,\; A=5, \;\lambda=0.1,\;x^2_{\min}=0$ gives us $K_O\leq 10^{-108}$ and a value of the left-hand side of \eqref{cfl2} larger than $10^{217}$!  To remedy to this issue, while maintaining our parameters $\lambda,\sigma$ and $A$, we have to allow only negative values for $R_0$. For instance, we may take $R_0\in\;]-50,-40[$. In order to set the CFL condition, we take moreover $\Delta t=1/1250$.
When implementing the real value in our scheme, we get at most the value $4$ for a time $T$ smaller than one quarter. After this, things are getting better and we have values much closer to zero, more precisely, whose orders are at most $10^{-3}$ (see figure \ref{fir}). Further, the approximation error of the real solution seems to be higher here, as represented in figure \ref{fea}.  With the preceding parameters, the left-hand side of \eqref{cfl2} is equal to 0.9009.
\begin{figure}
\raggedleft
\includegraphics[width=16.5cm,height=9cm]{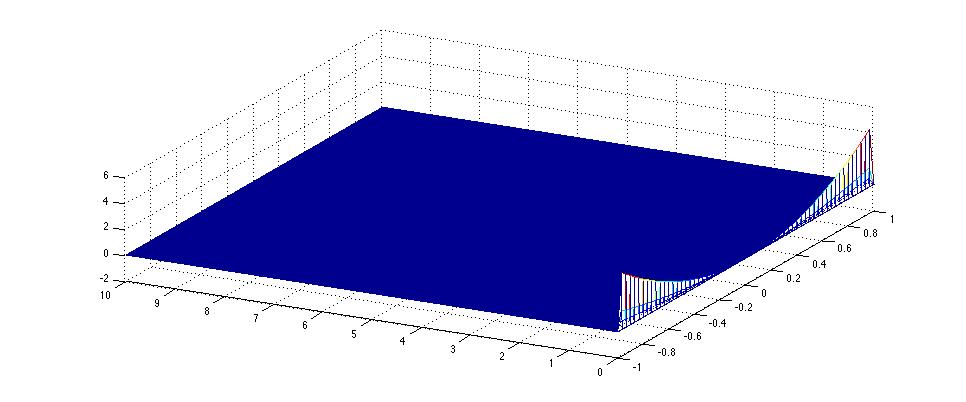}
\caption{Implementation of the real solution in our scheme for $B=1$.} 
\label{fir}
\end{figure}
\begin{figure}
\raggedleft
\includegraphics[width=16.5cm,height=9cm]{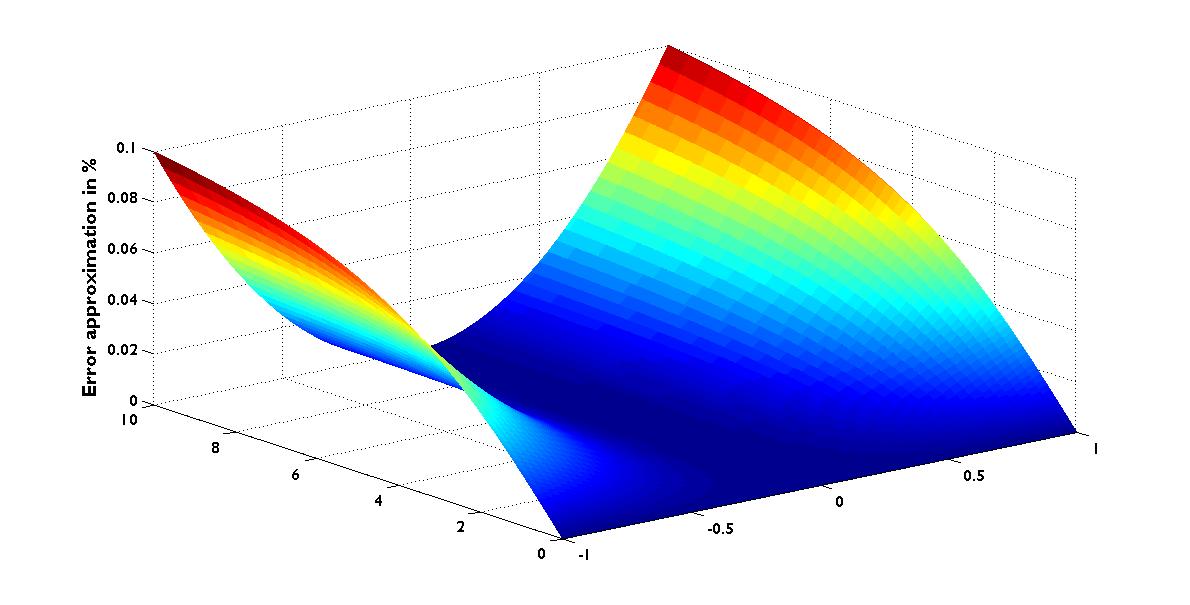}
\caption{Approximation error for $B=1$.} 
\label{fea}
\end{figure}

\subsection{Convex combinations of exponential utility functions}
In this subsection, we suppose that there exist $A_2>1>A_1>0$ and $\mu\in\;]0,1[$ such that 
\[
u(x)=\mu\big(1/A_1-\exp(-A_1x)\big)-(1-\mu)\exp\big(-A_2x).
\]
With this formulation of $u$, no well-known explicit formula for the associated value function (and hence for the solution of the associated auxiliary equation) can be given. Note that taking the corresponding convex combination of exponential value functions gives us only a supersolution of the corresponding HJB equation, since the functional on the left-hand side of \eqref{hjb} is subadditiv. Our goal in this section is to give an approximated value of the viscosity solution of $\eqref{hab}$. As discussed previously, we are going to use the explicit scheme to  achieve this. Let us start by finding a lower bound $K_O$ for $w_r$. To this end, we use inequalities \eqref{ieu'} and  and \eqref{ubd1} to infer
\begin{IEEEeqnarray*}{rCl}
w_r=\frac{-V_r}{B-V}&=&\frac{\E\Big[-u'\Big(\cR^{\xi^*}_T\Big)\Big]}{B-\E\Big[u\Big(\cR^{\xi^*}_T\Big)\Big]}\\
&\leq& \frac{-1- \E\Big[\exp(-A_1\Big(\cR^{\xi^*}_T\Big)\Big]}{B-V_2(T,X_0,R_0)}\\
&\leq& \frac{V_1(T,X_0,R_0)-(1+1/A_1)}{B-V_2(T,X_0,R_0)},           
\end{IEEEeqnarray*}
and in the case where $f(x)=\lambda x^2$, we get 
\[
w_r( T,X_0,R_0)=-\frac{1+\exp\bigg(-A_1R_0+X_0^2\sqrt{\frac{\lambda A_1^3\sigma^2}2}\coth\bigg(T\sqrt{\frac{A_1\sigma^2}{2\lambda}}\bigg)\bigg)}{B+\exp\bigg(-A_2R_0+X_0^2\sqrt{\frac{\lambda A_2^3\sigma^2}2}\coth\bigg(T\sqrt{\frac{A_2\sigma^2}{2\lambda}}\bigg)\bigg)}.
\]
For the sake of simplicity, take $B=1$ (consequently, we will have to take $\mu/A_1<1$ in order for $\log(B-u)$ to be well-defined), then we obtain with $$O=\;]\Delta t;T]\;\times\;]-X_{\max};X_{\max}[\;\times\;]0;R_{\max}[,$$
the following lower bound: 
\[
w_r( T,X_0,R_0)\leq-\frac{1+\exp(-A_1R_{\max})}{1+\exp\bigg(-A_2R_{\max}+X_{\max}^2\sqrt{\frac{\lambda A_2^3\sigma^2}2}\coth\bigg(T\sqrt{\frac{A_2\sigma^2}{2\lambda}}\bigg)\bigg)}=:-K_O.
\]
In the sequel we set:
$$O=\;]0.04,10]\;\times\;]-2,2[\;\times\;]0,20[,\; dr=0.8,\; dx=0.1,\; dt=0.013.$$ 
In Figure \ref{fig7},  the approximate value of the solution of \eqref{hjba} is displayed.
\begin{figure}
\raggedleft
\includegraphics[width=1.1\textwidth]{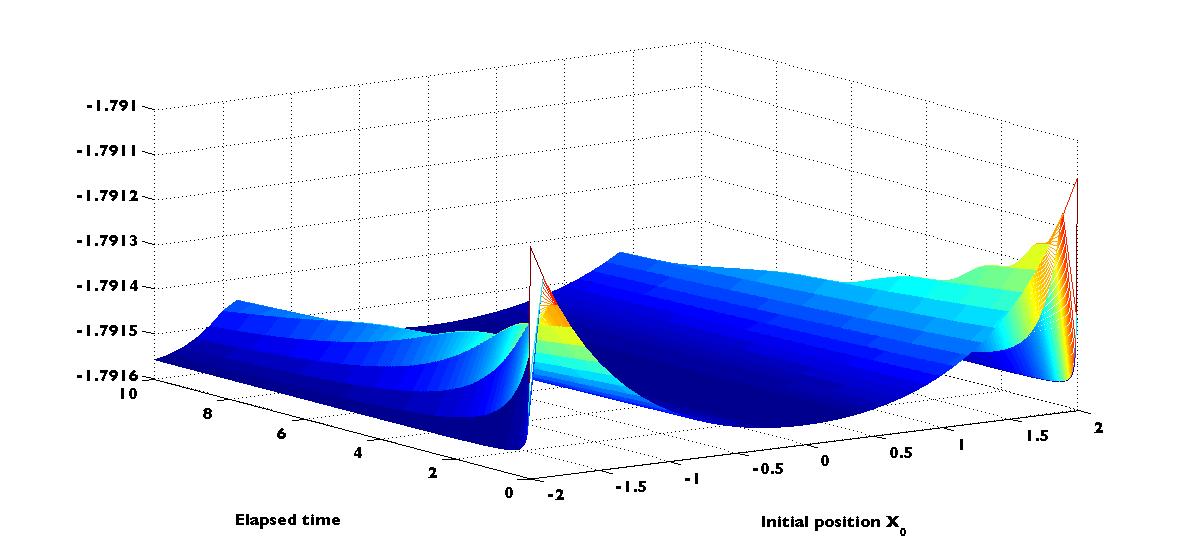}
\caption{Approximated value of the solution of \eqref{hjba}} 
\label{fig7}
\end{figure}
In Figure \ref{fig8}, we give an approximated representation of the value function of \eqref{hjb}. Note that the approximate displayed value function is concave for a fixed time when $x$ takes  values far enough from the boundaries (e.g., $x\in [-1.45;1.45]$), which is in concordance with Theorem \ref{v_r}.
\newpage
\begin{figure}
\raggedleft
\includegraphics[width=1.1\textwidth]{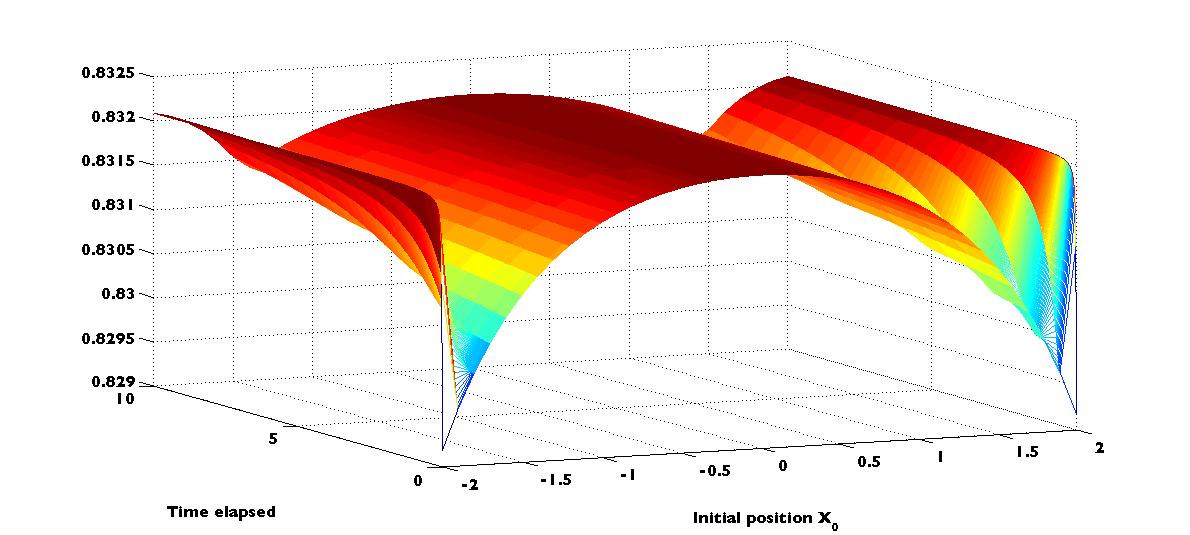}
\caption{Approximated value of the solution of \eqref{hjb}} 
\label{fig8}
\end{figure}

\bibliographystyle{plainnat}

\end{document}